\documentclass[11pt]{article}

\usepackage{xspace}
\usepackage{verbatim}
\usepackage{color}
\usepackage{graphicx}
\usepackage{tabularx}
\usepackage{amsthm,amsmath,amstext,amssymb,amsfonts}
\usepackage{nicefrac}
\usepackage{boxedminipage}
\usepackage{fullpage}
\newenvironment{mybox}
{
\center \noindent\begin{boxedminipage}{1.0\linewidth} }
{
\end{boxedminipage}\vspace{2ex}\\
}

\usepackage[pagebackref,colorlinks,linkcolor=blue,filecolor = blue,
citecolor = blue, urlcolor = blue]{hyperref}

\usepackage{textcomp,setspace}
\usepackage{a4wide}

\newcommand{\nfrac}{\nicefrac}

\newcommand{\iprod}[1]{\langle #1\rangle}

\newcommand{\norm}[1]{\lVert#1\rVert}

\definecolor{DSgray}{cmyk}{0,0,0,0.7}
\newcommand{\Authornote}[2]{{\small\textcolor{DSgray}{\sf$<<<${  #1: #2 }$>>>$}}}

\newtheorem{theorem}{Theorem}[section]
\newtheorem{claim}[theorem]{Claim}

\newtheorem{proposition}[theorem]{Proposition}
\newtheorem{lemma}[theorem]{Lemma}

\newtheorem{conjecture}[theorem]{Conjecture}

\newtheorem{hypothesis}[theorem]{Hypothesis}

\theoremstyle{definition}
\newtheorem{definition}[theorem]{Definition}

\newtheorem*{cor-a}{Corollary}
\newtheorem*{claim-a}{Claim}
\newtheorem*{lemma-a}{Lemma}
\newtheorem*{theorem-a}{Theorem}
\newtheorem*{fact-a}{Fact}
\newtheorem*{question-a}{Question}
\newtheorem*{conjecture-a}{}

\usepackage{prettyref}

\newrefformat{eq}{\textup{(\ref{#1})}}
\newrefformat{lem}{Lemma~\ref{#1}}
\newrefformat{thm}{Theorem~\ref{#1}}
\newrefformat{cha}{Chapter~\ref{#1}}
\newrefformat{sec}{Section~\ref{#1}}
\newrefformat{tab}{Table~\ref{#1} on page~\pageref{#1}}
\newrefformat{fig}{Figure~\ref{#1} on page~\pageref{#1}}
\newrefformat{def}{Definition~\ref{#1}}

\newcommand{\Esymb}{{\bf E}}

\newcommand{\Psymb}{{\bf P}}

\DeclareMathOperator*{\E}{\Esymb}

\DeclareMathOperator*{\ProbOp}{\Psymb}

\renewcommand{\Pr}{\ProbOp}

\newcommand{\R}{\mathbb{R}}

\newcommand{\cA}{\mathcal A}
\newcommand{\cB}{\mathcal B}

\newcommand{\cL}{\mathcal L}

\newcommand{\cQ}{\mathcal Q}

\newcommand{\ipr}[2]{\langle #1, #2\rangle}

\newcommand{\full}{1}    

\newcommand{\rnote}[1]{{\Authornote{Rajsekar}{#1}}}

\newcommand{\anote}[1]{{\Authornote{Aravindan}{#1}}}

\newcommand{\edge}{\{u, v\} \in E(G)}
\newcommand{\qpr}{\textsf{QP-Ratio}\xspace}
\newcommand{\nqpr}{\textsf{Normalized QP-Ratio}\xspace}

\newcommand{\bx}{{\boldsymbol{x}}}
\newcommand{\by}{{\boldsymbol{y}}}
\newcommand{\sdp}{{\sf sdp}}

\newcommand{\opt}{{\sf opt}}
\newcommand{\val}{{\sf val}}
\renewcommand{\epsilon}{\varepsilon}
\newcommand{\eps}{\epsilon}
\newcommand{\psd}{\succeq 0}

\newcommand{\ww}{\mathbf{w}}
\newcommand{\uu}{\mathbf{w}}
\newcommand{\vv}{\mathbf{v}}
\newcommand{\cc}{\mathbf{c}}

\newcommand{\vvt}{\tilde{\vv}}
\newcommand{\otilde}{\widetilde{O}}

\title{On Quadratic Programming with a Ratio Objective}
\date{}
\author{Aditya Bhaskara\thanks{Department of Computer Science, Princeton University, and Center for Computational Intractability. Supported by NSF CCF 0832797.  Email: \textsf{bhaskara@cs.princeton.edu}} \and Moses Charikar\thanks{Department of Computer Science, Princeton University, and Center for Computational Intractability. Supported by NSF CCF 0832797. Email: \textsf{moses@cs.princeton.edu}} \and Rajsekar Manokaran\thanks{Department of Computer Science, Princeton University, and Center for Computational Intractability. Supported by NSF CCF 0832797. Email: \textsf{rajsekar@cs.princeton.edu}} \and Aravindan Vijayaraghavan\thanks{Department of Computer Science, Princeton University, and Center for Computational Intractability. Supported by NSF CCF 0832797. Email: \textsf{aravindv@cs.princeton.edu}}}

\begin{document}
\maketitle
\begin{abstract}
Quadratic Programming (QP) is the well-studied problem of maximizing
over $\{-1,1\}$ values the quadratic form $\sum_{i \ne j} a_{ij} x_i
x_j$. QP captures many known combinatorial optimization problems,
and assuming the unique games conjecture, semidefinite programming
techniques give optimal approximation algorithms.
We extend this body of work by
initiating the study of Quadratic Programming problems where the
variables take values in the domain $\{-1,0,1\}$. 
The specific problems we study are
\begin{eqnarray*}
\qpr : \mbox{\ \ } \max_{\{-1,0,1\}^n} \frac{\sum_{i \not = j} a_{ij} x_i x_j}{\sum x_i^2},
&\mbox{\ and\ }&
\nqpr : \mbox{\ \ }  \max_{\{-1,0,1\}^n} \frac{\sum_{i \not = j} a_{ij} x_i x_j}{\sum d_i x_i^2}.\\
&& \mbox{where\ } d_i = \sum_j |a_{ij}| 
\end{eqnarray*}


These are natural relatives of several well studied problems (in fact
Trevisan introduced the latter problem as a stepping stone towards
a spectral algorithm for Max Cut Gain).
These quadratic ratio 
problems are good testbeds for both algorithms and complexity because 
the techniques used for quadratic problems for the $\{-1,1\}$ and $\{0,1\}$ 
domains do not seem to carry over to the $\{-1,0,1\}$ domain. We give 
approximation algorithms and evidence for the hardness of 
approximating these problems.

We consider an SDP relaxation obtained by adding constraints to the
natural eigenvalue (or SDP) relaxation for this problem. Using this,
we obtain an $\tilde{O}(n^{1/3})$ algorithm for QP-ratio.
%
We also obtain an $\tilde{O}(n^{1/4})$ approximation for bipartite graphs, and better algorithms 
for special cases.

As with other problems with ratio objectives (e.g. uniform sparsest
cut), it seems difficult to obtain inapproximability results based on
$\mathbf{P} \ne \mathbf{NP}$.
We give two results that indicate that \qpr is hard to
approximate to within any constant factor:
one is based on the assumption that random instances of Max $k$-AND
are hard to approximate, and
the other makes a connection to a ratio version of Unique Games.

There is an embarrassingly large gap between our upper bounds and lower bounds.
In fact, we give a natural distribution on instances of \qpr for which
an $n^\epsilon$ approximation (for small $\epsilon$)
seems out of reach of 
current techniques. 
On the one hand, this distribution presents a concrete barrier for
algorithmic progress.
On the other hand, it is a challenging question to develop lower bound
machinery to establish a hardness result of $n^{\epsilon}$ for this problem.

\end{abstract}
\newpage

\section{Introduction}
Semidefinite programming techniques have proved very useful for
quadratic optimization problems (i.e. problems with a quadratic
objective) over $\{0,1\}$ variables or $\{\pm 1\}$ variables. Such
problems admit natural SDP relaxations and beginning with the seminal
work of Goemans and Williamson \cite{gw}, sophisticated techniques
have been developed for exploiting these SDP relaxations to obtain approximation
algorithms.  For a large class of constraint satisfaction problems, a
sequence of exciting results\cite{kkmo,ko,kv} culminating in the work
of Raghavendra\cite{ragh}, shows that in fact, such SDP
based algorithms are optimal (assuming the Unique
Games Conjecture).

In this paper, we initiate a study of quadratic programming problems
with variables in $\{0,\pm 1\}$.  In contrast to
their well studied counterparts with variable values in $\{0,1\}$ or
$\{\pm 1\}$, to the best of our knowledge, such problems have not been
studied before.  These problems admit natural SDP relaxations similar
to problems with variable values in $\{0,1\}$ or $\{\pm 1\}$, yet we
know very little about how (well) these problems can be approximated.
We focus on some basic problems in this class: 
%
%
\begin{eqnarray}
\qpr : \mbox{\ \ } \max_{\{-1,0,1\}^n} \frac{\sum_{i \not = j} a_{ij} x_i x_j}{\sum x_i^2},
&\mbox{\ and\ }&
\nqpr : \mbox{\ \ }  \max_{\{-1,0,1\}^n} \frac{\sum_{i \not = j} a_{ij} x_i x_j}{\sum d_i x_i^2}.\label{eq:qpratio:defn}\\
&& \mbox{where\ } d_i = \sum_j |a_{ij}| \nonumber
\end{eqnarray}

Note that the numerator is the well studied quadratic programming
objective $\sum_{i < j} a_{i,j} x_i x_j$.  Ignoring the value of the
denominator for a moment, the numerator can be maximized by setting
all variables to be $\pm 1$.
However, the denominator term in the objective makes it worthwhile
to set variables to 0.  An alternate phrasing of the ratio-quadratic
programming problems is the following: the goal is to select a subset
of non-zero variables $S$ and assign them values in $\{\pm 1\}$ so as
to maximize the ratio of the quadratic programming objective
$\sum_{i<j \in S} a_{i,j} x_i x_j$ to the (normalized) size of $S$.

This problem is a variant of well studied problems: Eliminating 0 as a
possible value for variables gives rise to the problem of maximizing
the numerator over $\{\pm 1\}$ variables -- a well studied problem
with an $O(\log n)$ approximation~\cite{nemirovski, other-one, CW}. On the other hand, eliminating $-1$
as a possible value for variables (when the $a_{i,j}$ are
non-negative) results in a polynomial time solvable problem.  
Another closely related
problem to QP-Ratio is a {\em budgeted} variant where the goal is to
maximize the numerator (for the QP-Ratio objective) subject to the
denominator being at most $k$.  This is harder than QP-Ratio in the
sense that an $\alpha$-approximation for the budgeted version
translates to an $\alpha$-approximation for QP-Ratio (but not vice
versa).  The budgeted version is a generalization of $k$-Densest
Subgraph, a well known problem for which there is a huge gap between
current upper\cite{bccfv} and lower bounds\cite{khot,feige}. In this
paper, we chose to focus on the ``easier'' class of ratio problems.

Though it is a natural variant of well studied problems, QP-Ratio
seems to fall outside the realm of our current understanding on both
the algorithmic and inapproximability fronts.  One of the goals of our
work is to enhance (and understand the limitations of) the SDP toolkit
for approximation algorithms by applying it to this natural problem.
On the hardness side, the issues that come up are akin to those
arising in other problems with a ratio/expansion flavor, where
conventional techniques in inapproximability have been ineffective.


The \nqpr objective 
arose in recent work of Trevisan\cite{trevisan} on computing Max Cut
Gain using eigenvalue techniques.  The idea here is to use the
eigenvector to come up with a `good' partial assignment, and
recurse. Crucial to this procedure is a quantity called the
{\em GainRatio} defined for a graph; this is a special case of \nqpr where 
$a_{ij} = -1$ for edges, and 0 otherwise.


\subsection{Our results}\label{sec:our-results}
We first study mathematical programming relaxations for \qpr{}. The
main difficulty in obtaining such relaxations is imposing the
constraint that the variables take values $\{-1, 0, 1 \}$. 
Capturing this using convex constraints is the main
challenge in obtaining good algorithms for the problem.

We consider a semidefinite programming (SDP) relaxation obtained by
adding constraints to the natural eigenvalue relaxation, and round it
to obtain an $\otilde(n^{1/3})$ approximation algorithm.
An interesting special case is bipartite instances of \qpr{}, where the support
of $a_{ij}$ is the adjacency matrix of a bipartite graph
(akin to bipartite instances of quadratic programming, also known as the Grothendieck problem).
For bipartite instances, we obtain an $\tilde{O}(n^{1/4})$ approximation
and an almost matching SDP integrality gap of $\Omega(n^{1/4})$. 

Our original motivation to study quadratic ratio problems was the GainRatio problem
studied in Trevisan~\cite{trevisan}. We give a sharp contrast between the strengths of
different relaxations for the problem and disprove Trevisan's conjecture that
the eigenvalue approach towards Max Cutgain matches the bound achieved 
by an SDP-based approach\cite{CW}. See Section~\ref{sec:nqpr} for
details.

Complementing our algorithmic result for \qpr, we show hardness results for
the problem.
We first show that there is no PTAS for the problem assuming $P \ne NP$.
We also provide
evidence that it is hard to approximate to within any constant
factor.  We remark that current techniques seem insufficient to prove
such a result based on standard assumptions (such as $P \ne NP$) -- a
similar situation exists for other problems with a ratio objective
such as sparsest cut. 

In Section~\ref{sec:rand-k-and} we rule out constant factor
approximation algorithms for \qpr assuming that random instances of
$k$-AND are hard to distinguish from `well-satisfiable' instances.
This hypothesis was used as a basis to prove optimal hardness for the
so called $2$-Catalog problem (see \cite{feige}) and has proven
fruitful in ruling out $O(1)$-approximations for the densest subgraph
problem (see \cite{amm}). It is known that even very strong SDP
relaxations (in particular, $\Omega(n)$ rounds of the Lasserre
hierarchy) cannot refute this conjecture~\cite{tulsiani}.

%
  

We also show a reduction from {\textsf Ratio UG} (a ratio version of the well studied unique 
games problem), to \qpr.
We think that ratio version of Unique Games is an interesting problem
worthy of study that could shed light on the complexity of other 
ratio optimization questions.
The technical challenge in our reduction is to develop the required fourier-analytic machinery to tackle PCP-based reductions to ratio problems.

There is a big gap in the approximation guarantee of our algorithm and
our inapproximability results.  We suspect that the problem is in fact
hard to approximate to an $n^{\eps}$ factor for some $\eps>0$.  In
Section~\ref{sec:hiding}, we decribe a natural distribution over
instances which we believe are hard to approximate up to polynomial
factors. 
Our reduction from $k$-AND 
in fact generates hard instances of a similar structure
albeit ruling out only constant factor approximations.


\section{Algorithms for QP-Ratio}\label{sec:algo}
We start with the most natural relaxation for
\qpr{}~\eqref{eq:qpratio:defn} :
\[ \max \frac{\sum_{i,j} A_{ij} x_i x_j}{\sum_i x_i^2} \text{ subject to } x_i \in [-1,1] \] (instead of $\{0, \pm 1\}$). The solution to this is precisely the largest eigenvector of $A$ (scaled such that entries are in $[-1,1]$). However it is easy to construct instances for which this relaxation is bad: if $A$ were the adjacency matrix of a $(n+1)$ vertex star (with $v_0$ as the center of the star), the relaxation can cheat by setting $x_0=\frac{1}{2}$ and $x_i = \frac{1}{\sqrt{2n}}$ for $i \in [n]$ to give a gap of $\Omega(\sqrt{n})$ (the integer optimum is at most $1$). 

We show that SDP relaxations give more power in
expressing the constraints $x_i \in \{0, \pm 1\}$?
Consider the following relaxation:
\begin{align}
  \max \sum_{i,j} A_{ij} \cdot \iprod{\uu_i, \uu_j} ~~&\text{subject to } \sum_i \uu_i^2 = 1 \text{, and} \notag \\
  &|\iprod{\uu_i, \uu_j}| \le \uu_i^2 \text{ for all }i,j\label{eq:triang-ineq}
\end{align}
It is easy to see that this is indeed a relaxation: start with an integer solution $\{x_i\}$ with $k$ non-zero $x_i$, and set $\vv_i = (x_i/\sqrt{k}) \cdot \vv_0$ for a fixed unit vector $\vv_0$.

Without constraint ~\eqref{eq:triang-ineq}, the SDP relaxation is equivalent to the eigenvalue relaxation given above. 
Roughly speaking, equation~\eqref{eq:triang-ineq} tries to impose the constraint that non-zero vectors are of equal length.
In the example of the $(n+1)$-vertex star, this relaxation has value
equal to the true optimum.  In fact, for any instance with $A_{ij} \ge
0$ for all $i,j$, this relaxation is exact \cite{charikar}.
\ifnum\full<1
\footnote{We consider other SDP relaxations that can be writing by viewing the $\{0,\pm 1\}$ as a 3-alphabet CSP, and show a $\Omega(\sqrt{n})$ in the full version. It is interesting to see if lift and project methods starting with this relaxation can be useful.}
\fi

\ifnum\full=1
There are other natural relaxations one can write by viewing the $\{0,
\pm 1\}$ requirement like a 3-alphabet CSP. We consider one of
these in section~\ref{app:other-relax}, and show an $\Omega(n^{1/2})$ integrality gap. It is interesting to see if lift and project methods starting with this relaxation can be useful.
\fi

In the remainder of the section, we describe a simple $\otilde(n^{1/3})$ rounding algorithm, which shows that the additional constraints \eqref{eq:triang-ineq} indeed help.  We first describe an integrality gap of roughly $n^{1/4}$, as it highlights the issues that arise in rounding the SDP solution.

\subsection{Integrality gap instance}\label{sec:algo:gap}
Consider a complete bipartite graph on $L, R$, with $|L| = n^{1/2}$,
and $|R| = n$. The edge weights are set to $\pm 1$ uniformly at
random. Denote by $B$ the $n^{1/2} \times n$ matrix of edge weights
(rows indexed by $L$ and columns by $R$).  A
standard Chernoff bound argument shows 
\ifnum\full<1
(see the full version for a proof):
\fi
\begin{lemma}\label{lem:igap:lb}
With high probability over the choice of $B$, we have $\opt \le \sqrt{\log n} \cdot n^{1/4}$.
\end{lemma}

\ifnum\full=1
\begin{proof}
Let $S_1 \subseteq L$, $S_2 \subseteq R$ be of sizes $a$, $b$ respectively. Consider a solution in which these are the only variables assigned non-zero values (thus we fix some $\pm 1$ values to these variables). Let $\val$ denote the value of the numerator. By the Chernoff bound, we have
\[ \Pr[ \val \ge c \sqrt{ab} ] \le e^{-c^2/3}, \]
for any $c>0$. Now choosing $c = 10\sqrt{(a+b)\log n}$, and taking
union bound over all choices for $S_1, S_2$ and the assignment (there are $\binom{\sqrt{n}}{a} \binom{n}{b} 2^{a+b}$ choices overall), we get that w.p. at least $1 - 1/n^3$, no assignment with this choice of $a$ and $b$ gives $\val$ bigger than $\sqrt{ab (a+b) \log n}$.  The ratio in this case is at most $\sqrt{\log n \cdot \frac{ab}{a+b}} \le \sqrt{\log n} \cdot n^{1/4}$. Now we can take union bound over all possible $a$ and $b$, thus proving that $\opt
\leq n^{1/4}$ w.p. at least $1-1/n$.
\end{proof}
\fi

Let us now exhibit an SDP solution with value $n^{1/2}$. Let $\vv_1, \vv_2, \dots, \vv_{\sqrt{n}}$ be mutually orthogonal vectors, with each $\vv_i^2 = 1/2n^{1/2}$. We assign these vectors to vertices in
$L$. Now to the $j$th vertex in $R$, assign the vector $\uu_j$ defined by $\uu_j = \sum_{i} B_{ij} \frac{\vv_i}{\sqrt{n}}$.

\ifnum\full<1
It is easy to check that this assignment satisfies the SDP constraints , and attains a value $\Omega(n^{1/2})$. This gives a gap of $\Omega(n^{1/4})$. 

This gap instance can be seen as a collection of $n^{1/2}$ stars
(vertices in $L$ are the `centers'). In each `coordinate'
(corresponding to the orthogonal $\vv_i$), the assigment looks
like a star. $O(\sqrt{n})$ different coordinates allow us to
satisfy the constraints \eqref{eq:triang-ineq}.

\fi
\ifnum\full=1
It is easy to check that $\uu_j^2 = \sum_i \frac{\vv_i^2}{n} = \frac{1}{2n}$. Further, note that for any $i,
j$, we have (since all $\vv_i$ are orthogonal) $B_{ij} \iprod{\vv_i,
  \uu_j} = B_{ij}^2 \cdot \frac{\vv_i^2}{\sqrt{n}} =
\frac{1}{2n}$. This gives
$\sum_{i,j} B_{ij} \iprod{\vv_i, \uu_j} = n^{3/2} \cdot (1/2n) = n^{1/2}/2$.

From these calculations, we have $\forall i,j$, $|\vv_i \cdot \uu_j|
\le \uu_j^2$ (thus satisfying \eqref{eq:triang-ineq}; other
inequalities of this type are trivially satisfied). Further we saw
that $\sum_i \vv_i^2 + \sum_j \uu_j^2 = 1$. This gives a feasible solution of value
$\Omega(n^{1/2})$.  Hence the SDP has an $\widetilde{\Omega}(n^{1/4})$
integrality gap.

\paragraph{Connection to the star example.}
This gap instance can be seen as a collection of $n^{1/2}$ stars
(vertices in $L$ are the `centers'). In each `co-ordinate'
(corresponding to the orthogonal $\vv_i$), the assigment looks
like a star. $O(\sqrt{n})$ different co-ordinates allow us to
satisfy the constraints \eqref{eq:triang-ineq}.
\fi

This gap instance is bipartite.  In such instances it turns
out that there is a better rounding algorithm with a ratio $\otilde (n^{1/4})$ (Section~\ref{sec:bipartite-algo}).  Thus to bridge the gap between the algorithm and the integrality gap we need to better understand non-bipartite instances.

\subsection{An $O(n^{1/3})$ rounding algorithm}

Consider an instance of \qpr defined by $A_{(n\times n)}$.  Let $\uu_i$ be an optimal solution to the SDP, and let the objective value be denoted $\sdp$.  We will sometimes be
sloppy w.r.t. logarithmic factors in the analysis.

Since the problem is the same up to scaling the $A_{ij}$, let us assume that $\max_{i,j} |A_{ij}| = 1$.  There is a trivial solution which attains a value $1/2$ (if $i,j$ are indices with $|A_{ij}|=1$, set $x_i, x_j$ to be $\pm 1$ appropriately, and the rest of the $x$'s to $0$).  Now, since we are aiming for an $\otilde(n^{1/3})$ approximation, we can assume that $\sdp > n^{1/3}$.

As alluded to earlier (and as can be seen in the gap example),
the difficulty is when most of the contribution to $\sdp$ is from
non-zero vectors with very different lengths.  The idea of the
algorithm will be to move to a situation in which this does not
happen.  First, we show that if the vectors indeed have roughly equal length, we can round well. Roughly speaking, the algorithm uses the lengths $\norm{\vv_i}$ to determine whether to pick $i$, and then uses the ideas of~\cite{CW} (or
the earlier works of \cite{nemirovski, other-one}) applied to the vectors $\frac{\vv_i}{\norm{\vv_i}}$.
\begin{lemma}\label{lem:close-lengths}
Given a vector solution $\{ \vv_i \}$, with $\vv_i^2 \in [\tau/\Delta,
  \tau]$ for some $\tau>0$ and $\Delta>1$, we can round it to obtain
an integer solution with cost at least $\sdp / (\sqrt{\Delta} \log
n)$.
\end{lemma}
\begin{proof}
Starting with $\vv_i$, we produce vectors $\ww_i$ each of which is
either $0$ or a unit vector, such that
\[ \text{If }~ \frac{\sum_{i,j} A_{ij} \iprod{\vv_i, \vv_j}}{\sum_i \vv_i^2} = \sdp \text{, then } \frac{\sum_{i,j} A_{ij} \iprod{\ww_i, \ww_j}}{\sum_i \ww_i^2} \ge \frac{\sdp}{\sqrt{\Delta}}. \]
Stated this way, we are free to re-scale the $\vv_i$, thus we may assume $\tau =1$. Now note that once we have such $\ww_i$, we can throw away the zero vectors and apply the rounding algorithm of~\cite{CW} (with a loss of an $O(\log n)$ approximation factor), to obtain a $0, \pm 1$ solution with value at least $\sdp/ (\sqrt{\Delta} \log n)$.

So it suffices to show how to obtain the $\ww_i$. Let us set (recall we assumed $\tau =1$)
\[ \ww_i = \begin{cases} \vv_i / \norm{\vv_i} \text{, with prob. }\norm{\vv_i} \\ 0 \text{ otherwise} \end{cases} \]
(this is done independently for each $i$). Note that the probability of picking $i$ is proportional to the length of $\vv_i$ (as opposed to the typically used square lengths, \cite{cmm} say). 
Since $A_{ii}=0$, we have
\begin{equation}\label{eq:sdp-val}
\frac{ \E \big[ \sum_{i,j} A_{ij} \iprod{\ww_i, \ww_j} \big]}{ \E \big[ \sum_i \ww_i^2 \big]} = \frac{\sum_{i,j} A_{ij} \iprod{\vv_i, \vv_j}}{\sum_i |\vv_i|} \ge \frac{\sum_{i,j} A_{ij} \iprod{\vv_i, \vv_j}}{\sqrt{\Delta} \sum_i \vv_i^2} = \frac{\sdp}{\sqrt{\Delta}}. \end{equation}

The above proof only shows the existence of vectors
$\ww_i$ which satisfy the bound on the ratio. The proof can be made
constructive using the method of conditional expectations,
\ifnum\full<1
by setting the variables one by one, and use the fact that if $c,d >0$ and $\frac{a+b}{c+d} > \theta$, then either $\frac{a}{c} >\theta$ or $\frac{b}{d} > \theta$.\fi \ifnum\full=1 where we set variables one by one, i.e. we first decide whether to make $\ww_1$
to be a unit vector along it or the $0$ vector, depending on which maintains the
ratio to be $\ge \theta=\frac{\sdp}{\sqrt{\Delta}}$. Now, after fixing $\ww_1$, we fix $\ww_2$ similarly etc., 
while always maintaining the invariant that the ratio $\ge \theta$.

At step $i$, let us assume that $\ww_1,\dots,\ww_{i-1}$ have already been set to either unit vectors or zero vectors. 
Consider $\vv_i$ and let $\vvt_i=\vv_i/\norm{\vv_i}$. $\ww_i = \vvt_i$ w.p. $p_i= \norm{\vv_i}$ and $0$ w.p $(1-p_i)$.
\vspace{0.2in}

In the numerator, $B = \E[\sum_{j\ne i,k \ne i} a_{jk} \iprod{w_j,w_k}]$ is contribution from terms not involving $i$.
Also let $\cc_i=\sum_{k \ne i} a_{ik}\ww_k$ and let $\cc'_i=\sum_{j \ne i} a_{ji}\ww_j$. Then, from equation~\ref{eq:sdp-val}
\begin{equation*}
\theta \le \frac{\E[ \sum_{j,k} a_{jk} \iprod{\ww_j,\ww_k} ]}{E[ \sum_j |\ww_j|^2 ]} =
      \frac{ p_i \big( \iprod{\vvt_i, \cc_i}+\iprod{\cc'_i,\vvt_i}+B \big) + (1-p_i)B }{ p_i\big(1+\sum_{j\ne i} \norm{\ww_j}^2 \big) + (1-p_i)\big(\sum_{j\ne i} \norm{\ww_j}^2 \big) )}
\end{equation*}
Hence, by the simple fact that if $c,d$ are
positive and $\frac{a+b}{c+d} > \theta$, then either $\frac{a}{c} >
\theta$ or $\frac{b}{d} > \theta$, we see that either by setting $\ww_i=\vvt_i$ or $\ww_i=0$, we get value at least $\theta$.
\fi
\end{proof}
Let us define the `value' of a set of vectors $\{
\uu_i \}$ to be $\val := \frac{\sum A_{ij} \iprod{\uu_i, \uu_j}}{\sum_i \uu_i^2}$.
The $\vv_i$ we start will have $\val=\sdp$.
\begin{claim}\label{lem:len-lb}
We can move to a set of vectors such that (a) $\val$ is at least $\sdp/2$, (b) each non-zero vector $\vv_i$ satisfies $\vv_i^2 \ge 1/n$, (c) vectors satisfy \eqref{eq:triang-ineq}, and (d) $\sum_i \vv_i^2 \le 2$.
\end{claim}
The proof is by showing that very small vectors can either be enlarged
or thrown away 
\ifnum\full<1 (proof in full version)\fi
. 
\ifnum\full=1
\begin{proof} Suppose $0 <\vv_i^2 < 1/n$ for some $i$. If $S_i = \sum_{j} A_{ij} \vv_i \cdot \vv_j \leq 0$, we can set $\vv_i=0$ and improve the solution. Now if $S_i > 0$, replace $\vv_i$ by $\frac{1}{\sqrt{n}} \cdot \frac{\vv_i}{\norm{\vv_i}}$ (this only increases the value of $\sum_{i,j}A_{ij} \iprod{\vv_i, \vv_j}$), and repeat this operation as long as there are vectors with $\vv_i^2 < 1/n$. Overall, we would only have increased the value of $\sum_{i,j} A_{ij} \vv_i \cdot \vv_j$, and we still have $\sum_i \vv_i^2 \leq 2$. Further, it is easy to check that $|\iprod{\vv_i, \vv_j}| \le \vv_i^2$ also holds in the new solution (though it might not hold in some intermediate step above).
\end{proof}
\fi
The next lemma also
gives an upper bound on the lengths -- this is where the
constraints~\eqref{eq:triang-ineq} are crucial. 
\ifnum\full<1
It uses equation~\ref{eq:triang-ineq} to upper bound the contribution from each vector -- hence large vectors can not contribute much in total, since they are few in number (see the full version for details).
\fi
\begin{lemma}\label{lem:len-ub}
Suppose we have a solution of value $Bn^{\rho}$ and $\sum_i \vv_i^2 \le 2$. We can move to a solution with value at least $Bn^\rho /2$, and $\vv_i^2 < 16/n^{\rho}$ for all $i$.
\end{lemma}
\ifnum\full=1
\begin{proof}
Let $\vv_i^2 > 16/n^\rho$ for some index $i$. Since $|\iprod{ \vv_i, \vv_j}| \le \vv_j^2$, we have that for each such $i$,
\[ \sum_j A_{ij} \iprod{\vv_i, \vv_j} \le B \sum_j \vv_j^2 \le 2B\]
Thus the contribution of such $i$ to the sum $\sum_{i,j} A_{ij} \iprod{\vv_i, \vv_j}$ can be bounded by $m \times 4B$, where $m$ is the number of indices $i$ with $\vv_i^2 > 16/n^\rho$. Since the sum of squares is $\le 2$, we must have $m \le n^\rho/8$, and thus the contribution above is at most $B n^\rho/2$. Thus the rest of the vectors have a contribution at least $\sdp/2$ (and they have sum of squared-lengths $\le 2$ since we picked only a subset of the vectors)
\end{proof}
\fi
\begin{theorem}\label{thm:algo-main}
Suppose $A$ is an $n \times n$ matrix with zero's on the diagonal.  Then there exists a polynomial time $\otilde (n^{1/3})$ approximation algorithm for the \qpr{} problem defined by $A$.
\end{theorem}
\begin{proof}
As before, let us rescale and assume $\max{i,j} |A_{ij}| =1$.  Now if $\rho>1/3$, Lemmas~\ref{lem:len-lb} and \ref{lem:len-ub} allow us to restrict to vectors satisfying $1/n \le \vv_i^2 \le 4/n^\rho$, and using Lemma~\ref{lem:close-lengths} gives the desired $\otilde(n^{1/3})$ approximation; if $\rho<1/3$, then the trivial solution of $1/2$ is an $\otilde(n^{1/3})$ approximation.
\end{proof}

\subsection{The bipartite case} \label{sec:bipartite-algo} 
In this section, we prove the following theorem:
\begin{theorem}\label{thm:bipartite:algo}
When $A$ is bipartite (i.e. the adjacency matrix of a weighted bipartite graph), there is a (tight upto logarithmic factor) $\otilde(n^{1/4})$ approximation algorithm for \qpr~.
\end{theorem} 

Bipartite instances of \qpr{} can be seen as the ratio analog of the Grothendieck problem \cite{AN}.
The algorithm works by rounding the semidefinite program relaxation from section~\ref{sec:algo}.
As before, let us assume $\max_{i,j} a_{ij} =1$ and consider a solution to the SDP~\eqref{eq:triang-ineq}.  To simplify the notation, let $u_i$ and $v_j$ denote the vectors on the two sides of the bipartition.  Suppose the solution satisfies: 
\[ (1) \sum_{(i,j) \in E} a_{ij} \ipr{u_i}{v_j} \ge n^{\alpha}, \qquad (2) \sum_i u_i^2 = \sum_j v_j^2 = 1. \]

If the second condition does not hold, we scale up the vectors on the smaller side, losing at most a factor $2$.
Further, we can assume from Lemma~\ref{lem:len-lb} that the squared lengths $u_i^2, v_j^2$ are between $\frac{1}{2n}$ and $1$.  Let us divide the vectors $\{ u_i \}$ and $\{ v_j \}$ into $\log n$ groups based on their squared length.  There must exist two levels (for the $u$ and $v$'s respectively) whose contribution to the objective is at least $n^{\alpha} / \log^2 n$.\footnote{Such a clean division into levels can only be done in the bipartite case -- in general there could be negative contribution from `within' the level.}  Let $L$ denote the set of indices corresponding to these $u_i$, and $R$ denote the same for $v_j$.  Thus we have $\sum_{i \in L, j \in R} a_{ij} \ipr{u_i}{v_j} \ge n^{\alpha}/\log^2 n$.  We may assume, by symmetry that $|L| \le |R|$.  Now since $\sum_j v_j^2 \le 1$, we have that $v_j^2 \le 1/|R|$ for all $j \in R$.  Also, let us denote by $A_j$ the $|L|$-dimensional vector consisting of the values $a_{ij}$, $i \in L$.  Thus
\begin{equation}\label{eq:ai-lower-bound}
\frac{n^\alpha}{\log^2 n} \le \sum_{i \in L, j \in R} a_{ij} \ipr{u_i}{v_j} \le \sum_{i \in L, j \in R} |a_{ij}| \cdot v_j^2 \le \frac{1}{|R|} \sum_{j \in R} \norm{A_j}_1.
\end{equation}

We will construct an assignment $x_i \in \{+1, -1\}$ for $i \in L$ such that $\frac{1}{|R|} \cdot \sum_{j \in R} \big|\sum_{i \in L} a_{ij} x_i \big|$ is `large'.  This suffices, because we can set $y_j \in \{+1, -1\}$, $j \in R$ appropriately to obtain the value above for the objective (this is where it is crucial that the instance is bipartite -- there is no contribution due to other $y_j$'s while setting one of them).

\begin{lemma}\label{lem:assignment}
There exists an assignment of $\{+1, -1\}$ to the $x_i$ such that
\[ \sum_{j \in R} \big|\sum_{i \in L} a_{ij} x_i \big| \ge \frac{1}{24} \sum_{j \in R} \norm{A_j}_2 \]
Furthermore, such an assignment can be found in polynomial time.
\end{lemma}
\begin{proof}
The intuition is the following: suppose $X_i, i \in L$ are i.i.d. $\{+1, -1\}$ random variables.  For each $j$, we would expect (by random walk style argument) that $\E \big[ \big|\sum_{i \in L} a_{ij} X_i \big|\big]  \approx \norm{A_j}_2$, and thus by linearity of expectation,
$\E \Big[ \sum_{j \in R} \big| \sum_{i \in L} a_{ij} X_i \big| \Big] \approx \sum_{j \in R} \norm{A_j}_2$.
Thus the existence of such $x_i$ follows.  This can in fact be formalized:
\begin{equation}\label{eq:lemassignment}
\E\big[ \big|\sum_{i \in L} a_{ij} X_i \big|\big]  \ge \norm{A_j}_2 / 12
\end{equation}
This equation is seen to be true from the following lemma
\ifnum\full<1
whose proof in full version:
\fi
\begin{lemma}\label{lem:randomsums}
Let $b_1, \dots, b_n \in \mathbb{R}$ with $\sum_i b_i^2 =1$, and let $X_1, \dots, X_n$ be i.i.d. $\{+1, -1\}$ r.v.s.  Then
\[ \E[ |\sum_i b_i X_i| ] \ge 1/12. \]
\end{lemma}
\ifnum\full=1
\begin{proof}
Define the r.v. $Z := \sum_i b_i X_i$.  Because the $X_i$ are i.i.d. $\{+1,-1\}$, we have $\E[ Z^2] = \sum_i b_i^2 = 1$.  Further, $\E[Z^4] = \sum_i b_i^4 + 6\sum_{i<j} b_i^2 b_j^2 < 3 (\sum_i b_i^2)^2 = 3$.  Thus by Paley-Zygmund inequality,
\[ \Pr[ Z^2 \ge \frac{1}{4} ] \ge \frac{9}{16} \cdot \frac{(\E[Z^2])^2}{\E[Z^4]} \ge \frac{3}{16}. \]
Thus $|Z|\ge 1/2$ with probability at least $3/16 > 1/6$, and hence $\E[|Z|] \ge 1/12$.
\end{proof}
\fi
We can also make the above constructive.  Let r.v. $S:= \sum_{j \in R} \big| \sum_{i \in L} a_{ij} X_i \big|$.  It is a non-negative random variable, and for every choice of $X_i$, we have 
\[S \le \sum_{j \in R} \sum_{i \in L} |a_{ij}| \le L^{1/2} \sum_{j \in R} \norm{A_j}_2 \le n^{1/2} \E[S]\] 
Let $p$ denote $\Pr[S < \frac{\E[S]}{2}]$.  Then from the above inequality, we have
that $(1-p) \ge \frac{1}{2n^{1/2}}$.  Thus if we sample the $X_i$ say $n$ times (independently), we {\em hit}  an assignment with a large value of $S$ with high probability.
\end{proof}

\begin{proof}[Proof of Theorem~\ref{thm:bipartite:algo}.]
By Lemma~\ref{lem:assignment} and Eq~\eqref{eq:ai-lower-bound}, there exists an assignment to $x_i$, and a corresponding assignment of $\{+1, -1\}$ to $y_j$ such that the value of the solution is at least 
\[ \frac{1}{|R|} \cdot \sum_{j \in R} \norm{A_j}_2 \ge \frac{1}{|R|~|L|^{1/2}} \sum_{j \in R} \norm{A_j}_1 \ge \frac{n^{\alpha}}{|L|^{1/2} \log^2 n}. \qquad \text{[By Cauchy Schwarz]} \]
Now if $|L| \le n^{1/2}$, we are done because we obtain an approximation ratio of $O(n^{1/4} \log^2 n)$.  On the other hand if $|L| > n^{1/2}$ then we must have $\norm{u_i}_2^2 \le 1/n^{1/2}$.  Since we started with $u_i^2$ and $v_i^2 $ being at least $1/2n$ (Lemma~\ref{lem:len-lb}) we have that all the squared lengths are within a factor $O(n^{1/2})$ of each other.  Thus by Lemma~\ref{lem:close-lengths} we obtain an approximation ratio of $O(n^{1/4} \log n)$.
This completes the proof.
\end{proof}

\subsection{Algorithms for special cases}
\ifnum\full<1
We obtain better approximation algorithms for \qpr~ in restricted settings. We defer the proofs to the full version.
\begin{itemize}
\item When $A$ is positive semi-definite ($A \psd$), we can use the eigenvector to
obtain an $O(\log^2 n)$ approximation for \qpr . In a very recent independent work, \cite{nikhil} recently showed how to obtain a $O(\sqrt{\log n})$ approximation to \qpr~ when $A$ is psd.

\item We can also obtain a much better approximation algorithm for \qpr when the maximum value of the instance is large, say $\epsilon d_{max}$, where $d_{max} = \max_i \sum_i |a_{ij}|$. For \qpr instances $A$ with $OPT(A) \ge \eps d_{max}$, we can find a solution of value $e^{-O(1/\eps)} d_{max}$ using techniques from section~\ref{sec:nqpr}. Typically $d_{max} << n$, in which case the optimal solution might need to pick a smaller portion of the graph. Note that, the approximation factor in this case is independent of $n$.
\end{itemize}

\fi

\ifnum\full=1
\subsubsection{Poly-logarithmic approximations for positive semidefinite matrices}\label{app:psd-qpratio}
The MaxQP problem has a better approximation guarantee (of $2/\pi$) when $A$ is psd. Even for the QP-Ratio problem, we can do better in this case than for general $A$. In fact, it is easy to obtain a polylog$(n)$ approximation.

This proceeds as follows: start with a solution to the eigenvalue relaxation (call the value $\rho$). Since $A$ is psd, the numerator can be seen as $\sum_i (B_i x)^2$, where $B_i$ are linear forms. Now divide the $x_i$ into $O(\log n)$ levels depending on their absolute value (need to show that $x_i$ are not too small -- poly in $1/n$, $1/|A|_\infty$). We can now see each term $B_i x_i$ a sum of $O(\log n)$ terms (grouping by level). Call these terms $C_i^1, \dots, C_i^\ell$, where $\ell$ is the number of levels. The numerator is upper bounded by $\ell (\sum_i \sum_j (C_i^j)^2)$, and thus there is some $j$ such that $\sum_i (C_i^j)^2$ is at least $1/\log^2 n$ times the numerator. Now work with a solution $y$ which sets $y_i =x_i$ if $x_i$ is in the $j$th level and $0$ otherwise.
This is a solution to the ratio question with value at least $\rho/\ell^2$. Further, each $|y_i|$ is either $0$ or in $[\rho, 2\rho]$, for some $\rho$.

From this we can move to a solution with $|y_i|$ either $0$ or $2\rho$ as follows: focus on the numerator, and consider some $x_i \neq 0$ with $|x_i| < 2\rho$ (strictly). Fixing the other variables, the numerator is a convex function of $x_i$ in the interval $[-2\rho, 2\rho]$ (it is a quadratic function, with non-negative coefficient to the $x_i^2$ term, since $A$ is psd). Thus there is a choice of $x_i = \pm 2\rho$ which only increases the numerator. Perform this operation until there are no $x_i \neq 0$ with $|x_i| < 2\rho$. This process increases each $|x_i|$ by a factor at most $2$. Thus the new solution has a ratio at least half that of the original one. Combining these two steps, we obtain an $O(\log^2 n)$ approximation algorithm.

\subsubsection{Better approximations when the optimum is large.}
We can also obtain a much better approximation algorithm for \qpr when the maximum value of the instance is large, say $\epsilon d_{max}$, where $d_{max} = \max_i \sum_i |a_{ij}|$. For \qpr instances $A$ with $OPT(A) \ge \eps d_{max}$, we can find a solution of value $e^{-O(1/\eps)} d_{max}$ using techniques from section~\ref{sec:nqpr}.

This is because when all the degrees $d_i$ are roughly equal (say $\gamma d_{max}\le d_i \le d_{max}$ for some constant $\gamma>0$), then it is easy to check that an $O(\alpha)$ approximation to \nqpr~ (defined in section~\ref{sec:nqpr}) is an $O(\alpha/\gamma)$ approximation to the same instance of \qpr. Further, when $OPT(\qpr) \ge \eps d_{max}$, we can throw away vertices $i$ of degree $d_i< \frac{\eps}{2} d_{max}$ without losing in the objective. Hence, for a \qpr instance $A$ when $OPT(A) \ge \eps d_{max}$, we can find a solution to \qpr of value $e^{-O(1/\eps)} d_{max}$.

\subsection{Other Relaxations for \qpr{}}\label{app:other-relax}
For problems in which variables can take more than two values (e.g. CSPs with alphabet size $r>2$), it is common to use a relaxation where for every vertex $u$ (assume an underlying graph), we have variables $x_u^{(1)}, .., x_u^{(r)}$, and constraints such as $\iprod{x_u^{(i)}, x_u^{(j)}} = 0$ and $\sum_i \iprod{x_u^{(i)}, x_u^{(i)}} = 1$ (intended solution being one with precisely one of these variables being $1$ and the rest $0$).

We can use such a relaxation for our problem as well: for every $x_i$, we have three vectors $a_i, b_i$, and $c_i$, which are supposed to be $1$ if $x_i = 0, 1$, and  $-1$ respectively (and $0$ otherwise). In these terms, the objective becomes
\[ \sum_{i, j} A_{ij} \iprod{b_i, b_j} - \iprod{b_i, c_j} - \iprod{c_i, b_j} + \iprod{c_i, c_j} = \sum_{i,j} A_{ij} \iprod{b_i - c_i, b_j - c_j}. \]
The following constraints can be added
\begin{align}
\sum_i b_i^2 + c_i^2 = 1\\
\iprod{a_i, b_j}, \iprod{b_i, c_j}, \iprod{a_i, c_j} \ge 0 ~\text{for all }i, j\\
\iprod{a_i, a_j}, \iprod{b_i, b_j}, \iprod{c_i, c_j} \ge 0 ~\text{for all }i, j\\
\iprod{a_i, b_i} = \iprod{b_i, c_i} = \iprod{a_i, c_i} =0\\
a_i^2 + b_i^2 + c_i^2 = 1~\text{for all }i
\end{align}
Let us now see why this relaxation does not perform better than the one in \eqref{eq:triang-ineq}. Suppose we start with a vector solution $\uu_i$ to the earlier program. Suppose these are vectors in $\R^d$. We consider vectors in $\R^{n+d+1}$, which we define using standard direct sum notation (to be understood as concatenating co-ordinates). Here $e_i$ is a vector in $\R^n$ with $1$ in the $i$th position and $0$ elsewhere. Let $0_n$ denote the $0$ vector in $\R^n$.

We set (the last term is just a one-dim vector)
\begin{align*}
b_i = 0_n \oplus  \nfrac{\uu_i}{2} \oplus (\nfrac{|\uu_i|}{2}) \\ 
c_i = 0_n \oplus -\nfrac{\uu_i}{2} \oplus (\nfrac{|\uu_i|}{2}) \\
a_i = \sqrt{1-\uu_i^2} \cdot e_i \oplus 0_d \oplus (0)
\end{align*}
It is easy to check that $\iprod{a_i, b_j} = \iprod{a_i, c_j} = 0$, and $\iprod{b_i, c_j} = \nfrac{1}{4} \cdot (-\iprod{\uu_i, \uu_j} + |\uu_i||\uu_j|) \ge 0$ for all $i,j$ (and for $i=j,~\iprod{b_i, c_i}=0$). Also, $b_i^2 + c_i^2 = \uu_i^2 = 1-a_i^2$. Further, $\iprod{b_i, b_j} = \nfrac{1}{4} \cdot (\iprod{\uu_i, \uu_j} + |\uu_i||\uu_j|) \ge 0$.
Last but not least, it can be seen that the objective value is
\[ \sum_{i,j} A_{ij}\iprod{b_i-c_i, b_j-c_j} = \sum_{i,j} A_{ij} \iprod{\uu_i, \uu_j}, \]
as desired. Note that we never even used the inequalities~\eqref{eq:triang-ineq}, so it is only as strong as the eigenvalue relaxation (and weaker than the sdp relaxation we consider). 

Additional valid constraints of the form $a_i+b_i+c_i=v_0$ (where $v_0$ is a designated fixed vector) can be introduced -- however it it can be easily seen that these do not add any power to the relaxation.
\fi

\section{Normalized \qpr~} \label{sec:nqpr}

Given any symmetric matrix $A$, the normalized QP-Ratio problem aims to find the best $\{-1,0,1\}$ assignment which maximizes the following:
\begin{align}
\max_{{\bf x} \in \{-1,0,1\}^n}&\frac{\sum_{i \ne j} 2a_{ij} x_i x_j}{\sum_{i \ne j} |a_{ij}|(x_i^2+x_j^2)}\\
&=\frac{x^t A x}{\sum_i d_i x_i^2} \quad \text{ where } d_i=\sum_j |a_{ij}| \text{ are ``the degrees''} \notag
\end{align}

Note that when the degrees $d_i$ are all equal ($d_i=d \quad \forall i$), this is the same as \qpr~ upto a scaling. Though the two objectives have a very similar flavor, the normalized objective tends to penalize picking vertices of high degree in the solution.

This problem was recently considered by Trevisan~\cite{trevisan} in the special case when $A=-W(G)$ where $W(G)$ are the matrix of edge weights ($0$ if there is no edge) and called this quantity the \emph{GainRatio} of $G$.  
He gave an algorithm for Max Cut-Gain which uses GainRatio as a subroutine, based purely on an eigenvalue relaxation (as opposed to the SDP-based algorithm of \cite{CW}). 
\ifnum\full<1
\paragraph{Eigenvalue relaxation for Max Cut-Gain}

Consider the natural relaxation 
\begin{equation}\label{eq:eigrelax}
\max_{{\bf x} \in [-1,1]^n} \frac{x^t A x}{\sum_i d_i x_i^2}=\frac{2\sum_{i \ne j} a_{ij} x_i x_j}{\sum_{i \ne j} |a_{ij}|(x_i^2+x_j^2)}
\end{equation}
This is also the maximum eigenvalue of $D^{-1/2}AD^{1/2}$ where $D$ is the diagonal matrix of degrees. Trevisan \cite{trevisan} gave a randomized rounding technique which uses just threshold cuts to show that if the
eigenvalue is $\eps$, the GainRatio is at least $e^{-O(1/\eps)}$. 
His algorithm for GainRatio can also be adapted to give an algorithm for Normalized QP-Ratio with a similar guarantee. We give more details in the full version. 

\fi
\ifnum\full=1
His algorithm for GainRatio can also be adapted to give an algorithm for Normalized QP-Ratio with a similar guarantee. We sketch it below. 

\subsection{Algorithm based on \cite{trevisan}}

Consider the natural relaxation 
\begin{equation}\label{eq:eigrelax}
\max_{{\bf x} \in [-1,1]^n} \frac{x^t A x}{\sum_i d_i x_i^2}
\end{equation}
This is also the maximum eigenvalue of $D^{-1/2}AD^{1/2}$ where $D$ is the diagonal matrix of degrees. Trevisan \cite{trevisan} gave a randomized rounding technique which uses just threshold cuts to give the following guarantee.

\begin{lemma}\cite{trevisan}\label{lem:trev}
In the notation stated above, for every $\gamma>0$, there exists $c_1,c_2>0$ with $c_1 c_2 \le \gamma e^{1/\gamma}$, such that given any $\bx \in \R^n$ s.t. $\bx^t A \bx \geq \eps \bx^t D \bx$, outputs a distribution over discrete vectors $\{-1,0,1\}^n$ (using threshold cuts) with the properties:
\begin{enumerate}
\item $|c_1 \E Y_iY_j - x_ix_j | \le \gamma (x_i^2+x_j^2)$
\item $\E |Y_i| \le c_2 x_i^2$
\end{enumerate} 
\end{lemma}

\begin{proposition}
Given a \nqpr instance $A$ with value at least $\eps$ finds a solution $y \in \{-1,0,1\}^n$ of value $e^{-O(1/\epsilon)}$. 
\end{proposition}
\begin{proof}
For the eigenvalue relaxation equation~\ref{eq:eigrelax}, there is a feasible solution $\bx$ such that $\bx^t A \bx \ge \eps \bx^t D \bx$ where $D_ii = \sum_j |a_{ij}|$ and $D_{ij}=0$ for $i \ne j$. 
Now applying Lemma~\ref{lem:trev}, we have
\begin{align*}
\E[ a_{ij} Y_iY_j ] &\ge \frac{1}{c_1} \big( a_{ij}x_ix_j - \gamma|a_{ij}|(x_i^2 + x_j^2) \big) \\
\E[ \sum_{ij}  a_{ij} Y_i Y_j]  &\ge \frac{1}{c_1} \big( \bx^t A \bx - 2 \gamma \bx^t D \bx) \\
& \ge \frac{(\eps - 2 \gamma)}{c_1} (\bx^t D \bx)
\end{align*}
Also, $\E[\sum_i d_i |Y_i|] \le c_2 \bx^t D \bx$. Hence, there exists some vector $\by \in \{-1,0,1\}^n$ of value $(\eps - 2\gamma)/c_1c_2$, which shows what we need for sufficiently small $\gamma < \eps/2$. As in previous section~\ref{sec:algo}, this can also be derandomized using the method of conditional expectations (in fact, since the distribution is just over threshold cuts, it suffices to run over all $n$ threshold cuts to find the vector $\by$).
\end{proof}
\fi

\ifnum\full=1

\subsection{Eigenvalue relaxation for Max Cut-Gain}

As mentioned earlier Trevisan~\cite{trevisan} shows that if the
eigenvalue is $\eps$, the GainRatio is at least $e^{-O(1/\eps)}$. 
\fi He also conjectures that there could a better dependence: that the
GainRatio is at least $\eps / \log(1/\eps)$, whenever $\text{eigenval}
= \eps$. This would give an eigenvalue based algorithm which matches
the SDP-based algorithm of \cite{CW}.  We show that this conjecture is
false, and describe an instance for which eigenval is $\eps$, but the
GainRatio is at most $\exp(-1/\eps^{1/4})$. This shows that
the eigenvalue based approach is necessarily `exponentially' weaker
than an SDP-based one. Roughly speaking, SDPs are stronger because
they can enforce vectors to be all of equal length, while this cannot
be done in an eigenvalue relaxation. 
\ifnum\full=1
First, let us recall the
eigenvalue relaxation for Max CutGain
\[ \text{Eig} = \max_{x_u \in [-1,1]} \frac{\sum_{\edge} - w_{uv} x_u
  x_v}{\sum_{\edge} |w_{uv}| (x_u^2 + x_v^2)}. \] 
\fi

\paragraph{Description of the instance.}
In what follows, let us fix $\eps$ to be a small constant, and write
$M=1/\eps$ (thought of as an integer), and $m= 2/\eps$.

The vertex set is $V = V_1 \cup V_2 \cup \dots \cup V_m$, where $|V_i|
= M^i$. We place a clique with edge weight $1$ on each set
$V_i$. Between $V_i$ and $V_{i+1}$, we place a complete bipartite
graph with edge-weight $(1/2 + \eps)$.  We will call $V_i$ the $i$th
level. Thus the total weight of the edges in the $i$th level is
roughly $M^{2i}/2$, and the weight of edges between levels $i$ and
$(i+1)$ is $(1/2 + \eps) M^{2i+1}$.

\begin{lemma}\label{lem:gainratio:completeness}
There exist $x_i \in [-1,1]$ such that
$ \frac{ \sum_{\edge} -w_{uv} \cdot x_u x_v}{ \sum_{\edge} |w_{uv}| (x_u^2 +
 x_v^2)} = \Omega \big( \eps^2 \big).$
\end{lemma}
\ifnum\full<1
Setting $x_u = (-1)^i \eps^i$, where $i$ is the level of $u$ turns out to work (see the full version for the proof).
\fi
\ifnum\full=1
\begin{proof}
Consider a solution in which vertices $u$ in level $i$ have $x_u =
  (-1)^i \eps^{i}$. We have
\begin{align*}
  \frac{ \sum_{\edge} -w_{uv} \cdot x_u x_v}{ \sum_{\edge} x_u^2 +
    x_v^2} &= \frac{- N_0^2 \sum_{i=1}^m (M^{2i}/2) \eps^{2i} + N_0^2
    \sum_{i=1} ^{m-1} \eps^{2i+1} (\frac{1}{2} + \eps) M^{2i+1}}{
    N_0^2 \sum_{i=1}^m (M^{2i}/2) \cdot 2\eps^{2i}
    + N_0^2 \sum_{i=1}^{m-1} (\frac{1}{2}+\eps) M^{2m+1} (\eps^{2i} + \eps^{2i+2})}\\
  &\ge \frac{-\frac{m}{2} + (m-1)(\frac{1}{2}+\eps)}{\frac{3}{\eps}m} \qquad \text{(noting $M\eps = 1$)}\\
  &= \Omega(\eps^2) \qquad \text{(setting $m \approx \frac{2}{\eps}$)}
\end{align*}
\end{proof}

\fi
\vspace{0.02in}

Let us now prove an upper bound on the GainRatio of $G$.
Consider the optimal
solution $Y$. Let the fraction of vertices $u$ in level $i$ with non-zero $Y_u$ be $\lambda_i$. Of these, suppose
$(\frac{1}{2}+\eta_i)$ fraction have $Y_u = +1$ and
$(\frac{1}{2}-\eta_i)$ have $Y_u = -1$. It is easy to see that we may
assume $\eta_i$'s alternate in sign 
\ifnum\full<1
(as $i$ ranges from $1$ to $m$,
because if $\eta_i$ and $\eta_{i+1}$ are both positive, we can
increase the value of the solution by swapping the signs of all
$\eta_j$ for $1\leq j \leq i$). 
\fi
Thus, for convenience, we will let
$\eta_i$ denote the negated values for the alternate levels and treat
all $\eta_i$'s as positive.

With these parameters,
\ifnum\full<1 a simple calculation (see the full version for details) shows that the 
$\qquad \text{Numerator} = 2 \Big( \sum_{i=1}^m - M^{2i} \lambda_i^2 \eta_i^2 +
  (1+2\eps)\sum_{i=1}^{m-1} M^{2i+1} \lambda_i \lambda_{i+1} \eta_i
  \eta_{i+1} \Big)$.

\fi
\ifnum\full=1  we see that
\begin{eqnarray*}
  \text{Numerator} = \sum_{i=1}^m M^{2i} \lambda_i^2 \Big[ - \frac{1}{2}
  \big(\frac{1}{2} +\eta_i \big)^2 - \frac{1}{2} \big(\frac{1}{2} -\eta_i \big)^2 +
  \big(\frac{1}{2} +\eta_i\big) \big(\frac{1}{2}-\eta_i \big)\Big] \\ +
  \sum_{i=1}^{m-1} \big( \frac{1}{2} + \eps \big) M^{2i+1} \lambda_i
  \lambda_{i+1} \Big[\big(\frac{1}{2} + \eta_i\big)\big(\frac{1}{2} +
  \eta_{i+1}\big) + \big(\frac{1}{2} - \eta_{i}\big)\big( \frac{1}{2} -
  \eta_{i+1}\big) \\~~~~- \big( \frac{1}{2} + \eta_i\big) \big(\frac{1}{2} -
  \eta_{i+1}\big) - \big( \frac{1}{2} - \eta_i \big) \big(\frac{1}{2} + \eta_{i+1}
  \big) \Big]\\ = 2 \Big( \sum_{i=1}^m - M^{2i} \lambda_i^2 \eta_i^2 +
  (1+2\eps)\sum_{i=1}^{m-1} M^{2i+1} \lambda_i \lambda_{i+1} \eta_i
  \eta_{i+1} \Big)
\end{eqnarray*}
Hence the numerator is
\begin{equation}
\text{Numerator} = 2 \Big( \sum_{i=1}^m - M^{2i} \lambda_i^2 \eta_i^2 +
  (1+2\eps)\sum_{i=1}^{m-1} M^{2i+1} \lambda_i \lambda_{i+1} \eta_i
  \eta_{i+1} \Big)
\end{equation}
\fi
Note that the denominator is at least $\sum_i \lambda_i M^{2i}$ (there is a contribution
from every edge at least one end-point of which has $Y$ nonzero).
We will in fact upper bound the quantity $\text{Numerator} / 2\sum_i \lambda_i
\eta_i M^i$. This clearly gives an upper bound on the ratio we are
interested in (as the $\eta_i$ are smaller than $1/2$). Let us write
$\gamma_i = \lambda_i \eta_i$.  We are now ready to prove the theorem which implies the desired gap.  A simple inequality useful in the proof is the following (it follows from the well-known fact that the largest eigenvalue of the length $n$ path is $\cos(\frac{\pi}{n+1}) \approx 1 - \frac{1}{n^2}$):\\
\ifnum\full<1
$\qquad \text{$\forall n>1$ and $x_i \in \mathbb{R}$, } x_1^2 + x_2^2 + \dots x_n^2 \geq \big(1+\frac{1}{n^2}\big) (x_1 x_2
+ x_2 x_3 + \dots x_{n-1}x_n)$.
\fi
\ifnum\full=1
\begin{equation}\label{eq:useful-inequality}
\text{$\forall n>1$ and $x_i \in \mathbb{R}$, } x_1^2 + x_2^2 + \dots x_n^2 \geq \big(1+\frac{1}{n^2}\big) (x_1 x_2
+ x_2 x_3 + \dots x_{n-1}x_n)
\end{equation}
\fi
\begin{theorem}\label{thm:gainratio}
Let $\gamma_i \geq 0$ be real numbers in $[0,1]$, and let $\eps, M, m$
be as before. Then
\begin{equation}\label{expr1}
  \frac{ - \sum_{i=1}^m {\gamma_i^2 M^{2i}} + (1+2\eps)
    \sum_{i=1}^{m-1} {\gamma_i \gamma_{i+1} M^{2i+1}} }{ \sum_{i=1}^m
    {\gamma_i M^{2i}}} < \frac{1}{M^{\sqrt{m}/4}}
\end{equation}
\end{theorem}
\ifnum\full<1
We refer the reader to the full version for the full proof. The gap instance presented here is a weighted graph. In the full version, we show how we can construct an unweighted graph as well. 
\fi
\ifnum\full=1
\begin{proof}
Consider the numbers $\gamma_i M^i$ and let $r$ be the index where
it is maximized. Denote this maximum value by $D$.

{\em Claim.}  Suppose $1 \le j\le m$ and $j \not\in [r-\frac{m^{1/2}}{2}, r+\frac{m^{1/2}}{2}]$ and $\gamma_j M^j \ge \frac{D}{M^{\sqrt{m}/4}}$.  Then \eqref{expr1} holds.\\
Suppose first that $j > r+\frac{m^{1/2}}{2}$.  The numerator numerator of \eqref{expr1} is
at most $D^2 \times 3m$ while the denominator is at least $\gamma_j
M^{2j} > \frac{D}{M^{\sqrt{m}/4}} \times M^{j} >
\frac{D}{M^{\sqrt{m}/4}} \times M^{r +\frac{1}{2} \sqrt{m}} > D^2
\times M^{\sqrt{m}/4}$. (the last inequality is because $D < M^r$, since $\gamma_r <1$).
This implies that the ratio is at most $\frac{1}{M^{\sqrt{m}/4}}$, and hence \eqref{expr1} holds.
  
Next, suppose $j < r - \frac{m^{1/2}}{2}$.  This means,
since $\gamma_j <1$, that $\gamma_r < M^{-\sqrt{m}/4}$. Thus the
numerator of \eqref{expr1} is bounded from above by $D \times 3m$ as
  above, while the denominator is at least $\gamma_r M^{2r} =
  D^2/\gamma_r > D^2 M^{\sqrt{m}/4}$. Thus the ratio is at most
  $\frac{1}{M^{\sqrt{m}/4}}$, proving the claim.

  Thus for all indices $j \not\in [r- \frac{1}{2} \sqrt{m}, r+
  \frac{1}{2} \sqrt{m}]$ (let us call this interval $\mathcal{I}$), we
  have $\gamma_j M^j < \frac{D}{M^{\sqrt{m}/4}}$.  We thus split the
  numerator of \eqref{expr1} as
  \[ \Big(-\sum_{i \in \mathcal{I}} \gamma_i^2 M^{2i} + (1+2\eps)
  \cdot \sum_{i, (i+1) \in \mathcal{I}} \gamma_i \gamma_{i+1} M^{2i+1}
  \Big) + \mbox{remaining terms} \]

  Note that the part in the parenthesis is $\leq 0$ by suitable
  application of Eq.~\eqref{eq:useful-inequality}, while the remaining terms are each
  smaller than $\frac{D^2}{M^{\sqrt{m}/2}}$. Thus the numerator is $<
  m \frac{D^2}{M^{\sqrt{m}/2}}$.  Note that the denominator (of
  \eqref{expr1}) is at least $D^2$. These two together complete the
  proof of the theorem.
\end{proof}
\fi
\ifnum\full=1
\paragraph{Moving to an unweighted instance}
Now we will show that by choosing $N_0$ large enough (recall that we chose $V_i$ of size $N_0
M^i$), we can bound how far cuts are from expectation. Let us start
with a simple lemma.

\begin{lemma}
  Suppose $A$ and $B$ are two sets of vertices with $m$ and $n$
  vertices resp. Suppose each edge is added independently at random
  w.p. $(\frac{1}{2} + \eps)$. Then $$ \Pr \big[|\mbox{\# Edges} -
  (\frac{1}{2} + \eps) mn| > t \sqrt{mn} \big] < e^{-t^2/2}$$
\end{lemma}
\begin{proof}
  Follows from Chernoff bounds (concentration of binomial r.v.s)
\end{proof}

Next, we look at an arbitrary partitioning of vertices with
$\lambda_i$ and $\eps_i$ values as defined previously ($\lambda_i$ of
the $Y$'s nonzero, and $(\frac{1}{2} + \eps_i)$ of them of some sign).

Let us denote $n_i = N_0 \lambda_i M^i$. Between levels $i$ and $i+1$,
the probability that the number of edges differs from expectation by
$(n_i n_{i+1})^{1/3}$ is at most (by the lemma above) $e^{-(n_i
  n_{i+1})^{2/3}}$.  By choosing $N_0$ big enough (say $M^{10m}$) we
can make this quantity smaller than $e^{-12m}$ (since each $n_i$ is at
least $M^{9m}$). Thus the probability that the sum of the `errors'
over the $m$ levels is larger than $\text{Err} = \sum (n_i
n_{i+1})^{1/3}$ is at most $me^{-12m}$.

The total number of vertices in the subgraph (with $Y$ nonzero) is at
most $N_0 M^m$. Thus the number of cuts is $2^{N_0 M^m} <
e^{11m}$. Thus there exists a graph where none of the cuts have sum of
the `errors' as above bigger than $\text{Err}$.

Now it just remains to bound $\frac{\sum (n_i n_{i+1})^{1/3}}{\sum
\lambda_i N_0^2 M^{2i}}$. Here again the fact that $N_0$ is big comes
to the rescue (there is only a $N_0^{2/3}$ in the numerator) and hence
we are done.
\fi
\section{Hardness of Approximating QP-Ratio}
\newcommand{\fnlin}{f^{\neq 1}}
\newcommand{\xalpha}{\alpha}
\newcommand{\xbeta}{\beta}
\newcommand{\xdelta}{\delta}
\newcommand{\xmu}{\mu}
\newcommand{\xeta}{\eta}
\newcommand{\sselc}{\textsf{SSE Label Cover}~}
\newcommand{\flin}{f^{=1}}
\newcommand{\ql}{\mathcal{Q}_{\mathcal{L}}}
\newcommand{\ff}{\widehat{f}}
\newcommand{\qpinter}{\textsf{QP-Intermediate}}
\newcommand{\sseug}{\textsf{SSE UG~}}
\newcommand{\del}{\delta}
\newcommand{\num}{\mathsf{num}}
\newcommand{\epslb}{\gamma}
\newcommand{\ratioug}{\textsf{Ratio UG~}}

Given that our algorithmic techniques give only an $n^{1/3}$ approximation in general, and the natural relaxations do not seem to help, it is natural to ask how hard we expect the problem to be.  Our results in this direction are as follows: we show that the problem is APX-hard (i.e., there is no PTAS unless $P=NP$).  Next, we show that there cannot be a constant factor approximation assuming that Max $k$-AND is
hard to approximate `on average' (related assumptions are explored in~\cite{feige}).  Our reduction therefore gives a
(fairly) natural {\em hard distribution} for the \qpr problem.

\subsection{Candidate Hard Instances}\label{sec:hiding}

As can be seen from the above, there is an embarrassingly large gap between our
upper bounds and lower bounds.  We attempt to justify this by describing
a natural distribution on instances we do not know how to approximate
to a factor better than $n^{\delta}$ (for some fixed $\delta > 0$).

Let $\cal{G}$ denote a bipartite random graph with vertex sets $V_L$ of size $n$
and $V_R$ of size $n^{2/3}$, left degree $n^{\del}$ for some small $\del$
(say 1/10) [i.e., each edge between $V_L$ and $V_R$ is picked i.i.d. with prob. $n^{-(9/10)}$].
Next, we pick a random (planted) subset $P_L$ of $V_L$ of size $n^{2/3}$ and random assignments
$\rho_L : P_L \mapsto \{+1, -1\}$ and $\rho_R : V_R \mapsto \{+1, -1\}$.  For an edge
between $i \in P_L$ and $j \in V_R$, the weight $a_{ij} := \rho_L(i) \rho_R(j)$.  For all
other edges we assign $a_{ij} = \pm 1$ independently at random.

The optimum value of such a {\em planted} instance is roughly $n^{\del}$, because
the assignment of $\rho_L, \rho_R$ (and assigning $0$ to $V_L \setminus P_L$)
gives a solution of value $n^{\del}$. 
However, for $\del < 1/6$, we do not know how to find such a planted assignment: simple
counting and spectral approaches do not seem to help. 

Making progress on such instances
appears to be crucial to improving the algorithm or the hardness results.
We remark that the instances produced by the reduction from Random $k$-AND are in fact
similar in essence.  We also note the similarity to other problems which are beyond current techniques, such as the Planted Clique and Planted Densest Subgraph problems~\cite{bccfv}.

\newcommand{\vthet}{\vartheta}

\subsection{Reduction from Random $k$-AND}
\label{sec:rand-k-and}
We start out by quoting the assumption we use.
\begin{conjecture}[Hypothesis $3$ in \cite{feige}]\label{conj:kand} For some constant $c > 0$, for every $k$, there is
  a $\Delta_0$, such that for every $\Delta > \Delta_0$, there is no
  polynomial time algorithm that, on most $k$-AND formulas with
  $n$-variables and $m = \Delta n$ clauses, outputs
  \texttt{`typical'}, but never outputs \texttt{`typical'} on
  instances with $m / 2^{c \sqrt{k}}$ satisfiable clauses.
\end{conjecture}

The reduction to \qpr{} is then as follows: Given a $k$-AND instance
on $n$ variables $X = \{x_1, x_2, \ldots x_n\}$ consisting of $m$
clauses $C = \{C_1, C_2, \ldots C_m\}$, and a parameter $0 < \alpha <
1$, let $A = \{a_{ij}\}$ denote the $m \times n$ matrix such that
$a_{ij}$ is $1/m$ if variable $x_j$ appears as is in clause $C_i$,
$a_{ij}$ is $-1/m$ if it appears negated and $0$ otherwise.

Let $f:X \rightarrow \{-1,0,1\}$ and $g:C \rightarrow \{-1,0,1\}$ denote functions which are supposed to correspond to assignments. Let $\mu_f = \sum_{i \in [n]}
|f(x_i)| / n$ and $\mu_g = \sum_{j \in m} |g(C_j)| / m$.  Let
\begin{equation}
\vthet(f,g) = \frac{\sum_{ij} a_{ij} f(x_i) g(C_j)}{\alpha \mu_f +   \mu_g}.  \label{eq:kand}
\end{equation}

Observe that if we treat $f(), g()$ as variables, we obtain an instance of QP-Ratio 
\ifnum\full<1
(we describe how get rid of the weighting in the denominator in the full version)
\fi
\ifnum\full=1[as
described, the denominator is {\em weighted}; we need to replicate the variable set $X$ $\alpha \Delta$ times (each copy has same set of neighbors in $C$)
in order to reduce to an unweighted instance -- see 
Appendix~\ref{app:kand:instance} 
for details]. \fi 
We pick $\alpha = 2^{-c\sqrt{k}}$ and $\Delta$ a large
enough constant so that Conjecture~\ref{conj:kand} and
Lemmas~\ref{lem:expansion} and \ref{lem:kand:adv} hold. The completeness follows from the natural assignment 
\ifnum\full<1 (proof in full version) \fi .
\begin{lemma}[Completeness] If the $k$-AND instance is such that an
  $\alpha$ fraction of the clauses can be satisfied,
  then there exists function $f$, $g$ such that $\theta$ is at least
  $k/2$.
\end{lemma}
\ifnum\full=1
\begin{proof}  Consider an assignment that satisfies an $\alpha$
  fraction of the constraints.  Let $f$ be such that $f(x_i) = 1$ if
  $x_i$ is true and $-1$ otherwise.  Let $g$ be the indicator of (the
  $\alpha$ fraction of the) constraints that are satisfied by the
  assignment. Since each such constraint contributes $k$ to the sum in the numerator, the numerator is
  at least $\alpha k$ while the denominator $2 \alpha$.
\end{proof}
\fi
\paragraph{Soundness:}
We will show that for a typical random $k$-AND instance (i.e.,
with high probability), the matrix $A$ is such that the maximum value $\vthet(f,g)$ can take is at most $o(k)$.

\newcommand{\qmax}{\vthet_{max}}
Let the maximum value of $\vthet$ obtained be $\qmax$. We first note
that there exists a solution $f,g$ of value $\qmax/2$ such that the equality 
$\alpha \mu_f= \mu_g$ holds\footnote{For instance, if $\alpha \mu_f >
  \mu_g$, we can always pick more constraints such that the numerator
  does not decrease (by setting $g(C_j)= \pm 1$ in a greedy way so as to not decrease the numerator) till
  $\mu_{g'}= \alpha \mu_f $, while losing a factor $2$. Similarly for $\alpha \mu_f < \mu_g$, we pick more variables.} -- so we only need consider such assignments.

Now, the soundness argument is
two-fold: if only a few of the vertices $(X)$ are picked ($\mu_f <
\frac{\alpha}{400}$) then the expansion of small sets guarantees that
the value is small \ifnum\full=1(even if each picked edge contributes $1$)\fi. On the
other hand, if many vertices (and hence clauses) are picked, then we claim that for every assignment to the variables
(every $f$), only a small fraction ($2^{-\omega(\sqrt{k})}$) of the
clauses contribute more than $k^{7/8}$ to the numerator.

The following lemma handles the first case when $\mu_f < \alpha/ 400$ \ifnum\full<1 
(proof in full version)\fi .
\begin{lemma}\label{lem:expansion}  Let $k$ be an integer, $0 <\del <1$, and $\Delta$ be large enough.  If we choose a bipartite graph with vertex sets $X, C$ of sizes $n, \Delta n$ respectively and degree $k$ (on the $C$-side) uniformly at random, then w.h.p., for every $T \subset X, S \subset C$ with $|T| \le n\alpha/400$ and $|S| \le \alpha |T|$, we have $|E(S, T)| \le \sqrt{k} |S|$.
\end{lemma}
\ifnum\full=1
\begin{proof} Let $\mu:= |T|/|X|$ (at most $\alpha/400$ by choice), and $m = \Delta n$.  Fix a subset $S$ of $C$ of size $\alpha \mu m$ and a subset $T$ of $X$ of size $\mu n$.  The expected number of edges  between $S$ and $T$ in $G$ is $\E [ E(S,T)] = k \mu \cdot |S|$.  Thus, by Chernoff-type bounds (we use only the {\em upper tail}, and we have negative correlation here),
\[\Pr[ E(S,T) \geq \sqrt{k}|S| ] \leq \exp \big(- \frac{(\sqrt{k}|S|)^2}{k\mu \cdot |S|} \big) \leq \exp\left(- \alpha m / 10 \right) \]

The number of such sets $S, T$ is at most $ 2^n \times \sum_{i =
    1}^{\alpha^2 m/400} \binom{m}{i} \le 2^n 2^{H(\alpha^2/400) m} \le
  2^{n + \alpha m/20}.$ Union bounding and setting $m/n > 20/\alpha$
  gives the result.
\end{proof}
\fi
Now, we need to bound $\vthet(f,g)$ for solutions such that $\alpha \mu_f = \mu_g \geq \alpha^2/400$. We use the following lemma about random instances of $k$-AND \ifnum\full<1 
(proof in full version)\fi . 
\begin{lemma} \label{lem:kand:adv} For large enough $k$ and $\Delta$, a random $k$-AND
  instance with $\Delta n$ clauses on $n$ variables is such that: for
  any assignment, at most a $2^{\frac{-k^{3/4}}{100}}$ fraction of the clauses have more than
$k/2 + k^{7/8}$ variables `satisfied' [i.e. the variable takes the value dictated by the AND clause] w.h.p.
\end{lemma}
\ifnum\full=1 \begin{proof} Fix an assignment to the variables $X$.  For a single
  random clause $C$, the expected number of variables in the clause
  that are satisfied by the assignment is $k/2$.  Thus, the
  probability that the assignment satisfies more than $k/2(1 +
  \delta)$ of the clauses is at most $\exp(-\delta^2 k/20)$. Further, each $k$-AND clause is chosen independently at random. Hence, by setting
  $\delta = k^{-\nfrac{1}{8}}$ and taking a union bound over all the $2^n$ assignments gives the result (we again use the fact that $m \gg n / \alpha$).
\end{proof}
Lemma~\ref{lem:kand:adv} shows that for every $\{\pm 1 \}^n$ assignment to the variables $x$, at most $2^{-\omega(\sqrt{k})}$ fraction of the clauses contribute more than $2k^{7/8}$ to the numerator of $\vthet(f,g)$.  We can now finish the proof of the soundness part above. \fi

\begin{proof}[Proof of Soundness.]
Lemma~\ref{lem:expansion} shows that when $\mu_f < \alpha/400$, $\vthet(f,g) = O(\sqrt{k})$.
For solutions such that $\mu_f > \alpha/400$, i.e., $\mu_g \ge \alpha^2/400 = 2^{-2\sqrt{k}}/400$, by Lemma~\ref{lem:kand:adv} at most $2^{-\omega(\sqrt{k})} ~ (\ll \mu_g/k)$ fraction of the constraints contribute more than $k^{7/8}$ to the numerator.  Even if the contribution is $k$ [the maximum possible] for this small fraction, the value $\vthet(f,g) \le O(k^{7/8})$.
\end{proof}
This shows a gap of $k$ vs $k^{7/8}$ assuming Hypothesis~\ref{conj:kand}.  Since we can pick $k$ to be arbitrarily large, we can conclude that \qpr{} is hard to approximate to any constant factor.

\subsection{Reductions from Ratio versions of CSPs}
This section is inspired by the proof of~\cite{qp-hard}, who show that Quadratic Programming (QP) is hard to approximate by giving a reduction from Label Cover to \qpr{}. 

Here we ask: is there a reduction from a {\em ratio version} of Label Cover to \qpr{}?  For this to be useful we must also ask: is the (appropriately defined) ratio version of Label Cover hard to approximate?  The answer to the latter question is yes [see 
\ifnum\full<1 the full version \fi \ifnum\full=1 section~\ref{app:towardsnp} \fi
for details and proof that Ratio-LabelCover is hard to approximate to any constant factor].  Unfortunately, we do not know how to reduce from Ratio-LabelCover.  However, we present a reduction starting from a ratio version of {\em Unique Games}.  We do not know if \ratioug{} is hard to approximate for the parameters we need.  While it seems related to Unique Games with Small-set Expansion~\cite{rs}, a key point to note is that we do not need `near perfect' completeness, as in typical UG reductions. 

We hope the Fourier analytic tools we use to analyze the ratio objective could find use in other PCP-based reductions to ratio problems.   Let us now define a ratio version of Unique Games, and a useful {\em intermediate} QP-Ratio problem.

\begin{definition}[\ratioug{}]
Consider a unique label cover instance $\mathcal{U}\big(G(V,E),[R],\{\pi_e | e \in E\}\big)$. The value of a partial labeling $L:V \rightarrow [R]\cup \{\bot\}$ (where label $\bot$ represents it is unassigned) is
\begin{displaymath}
val(L)=\frac{\vert \{(u,v) \in E| \pi_{u,v}(L(u))=L(v)\}\vert}{\vert \{v \in V| L(v) \ne \bot \}\vert}
\end{displaymath}
The $(s,c)$-\ratioug{} problem is defined as follows: given $c>s>0$ (to be thought of as constants), and an instance $\mathcal{U}$ on a regular graph $G$, distinguish between the two cases:
\begin{itemize}
\item {\bf YES:} There is a partial labeling $L:V\rightarrow [R] \cup \{\bot\}$, such that $val(L)\ge c$.
\item {\bf NO:} For every partial labeling $L:V\rightarrow [R] \cup \{\bot\}$, $val(L)< s$.
\end{itemize}
\end{definition}

\ifnum\full=1
The main result of this section is a reduction from $(s,c)$-\ratioug{} to QP-ratio.  We first introduce the following {\em intermediate} problem:
\fi
\begin{definition}
\qpinter{}: Given $A_{(n \times n)}$ with $A_{ii} \le 0$
, maximize $\frac{x^T A x}{\sum_i |x_i|}$ s.t. $x_i \in [-1,1]$.
\end{definition}

Note that $A$ is allowed to have diagonal entries (albeit only
negative ones)
\ifnum\full=1, and that the variables are allowed to take values in the {\em interval}
$[-1,1]$ \fi
. 
\ifnum \full<1 It will suffice to consider reductions to \qpinter{} (please refer to the full version for details). \fi

\ifnum\full=1
\begin{lemma}\label{lem:qpinter-qpr}
Let $A$ define an instance of \qpinter{} with optimum value $\opt_1$.  There exists an instance $B$ of \qpr{} on $(n \cdot m)$ variables, with $m \le \max\{\frac{2\norm{A}_1}{\eps}, 2n\}+1$, and the property that its optimum value $\opt_2$ satisfies $\opt_1 - \eps \le \opt_2 \le \opt_1+\eps$.  [Here $\norm{A}_1 = \sum_{i,j} |a_{ij}|$.]
\end{lemma}
\fi
\ifnum\full=1
\begin{proof}
The idea is to view each variable as an average of a large number (in this case, $m$) of new variables: thus a fractional value for $x_i$ is `simulated' by setting some of the new variables to $\pm 1$ and the others zero.  See Appendix~\ref{app:qpinter-qpr} for details.
\end{proof}
Thus from the point of view of approximability, it suffices to consider \qpinter{}.  We now give a reduction from \ratioug{} to \qpinter{}.
\fi
\ifnum\full<1
Now given an instance $\Upsilon=(V,E,\Pi)$ of \ratioug, with alphabet $[R]$ and a regular graph $(V,E)$, we associate $2^R$ variables to each vertex, which are denoted $f_u(x)$, indexed by $x \in \{-1,1\}^R$.  The intended solution to each vertex is either the long code corresponding to the label, or $f_u=0$ (for each $x$).  Now,
\begin{itemize}
	\item Note that Fourier coefficients ($\ff_u(S) = \E_x [\chi_S(x) f_u(x)]$) are linear forms in the variables $f_u(x)$.
	\item For $(u,v) \in E$, define $T_{uv} = \sum_i \ff_u(\{i\})\ff_v(\{ \pi_{uv}(i)\})$.  [This term is $1$ if an edge is satisfied]
	\item For $u \in V$, define $L(u) = \sum_{S:|S|\neq 1} \ff_u(S)^2$. [This term penalizes $f_u$ that are not dictators]
\end{itemize}
The instance of \qpinter{} we consider is (here $\norm{f_u}_1$ denotes $\E_x[|f_u(x)|]$)
\[ \cQ := \max ~\frac{\E_{(u,v)\in E} T_{uv} - \eta \E_u L(u)}{\E_u |f_u|_1}, \text{ where $\eta$ will be picked large enough.}\]
\fi

\ifnum\full=1
\begin{mybox}
  \smallskip
  \noindent\textbf{Input}: An instance $\Upsilon = (V, E, \Pi)$ of \ratioug, with alphabet $[R]$.\\
  \noindent \textbf{Output}: A \qpinter{} instance $\cQ$ with number of variables $N=|V| \cdot 2^R$.\\
  \noindent\textbf{Parameters}: 
$\eta := 10^6 n^7 2^{4R}$\\
  \smallskip

  \noindent\textbf{Construction}:
  \begin{itemize}
  \item For every vertex $u \in V$, we have $2^R$ variables, indexed by $x \in \{-1,1\}^R$. We will denote these by $f_u(x)$, and view $f_u$ as a function on the hypercube $\{-1,1\}^R$.
	\item Fourier coefficients (denoted $\ff_u(S) = \E_x [\chi_S(x) f_u(x)]$) are linear forms in the variables $f_u(x)$.
	\item For $(u,v) \in E$, define $T_{uv} = \sum_i \ff_u(\{i\})\ff_v(\{ \pi_{uv}(i)\})$.
	\item For $u \in V$, define $L(u) = \sum_{S:|S|\neq 1} \ff_u(S)^2$.
	\item The instance of \qpinter{} we consider is
\[ \cQ := \max ~\frac{\E_{(u,v)\in E} T_{uv} - \eta \E_u L(u)}{\E_u |f_u|_1}, \]
	where $|f_u|_1$ denotes $\E_x [|f_u(x)|]$.
\end{itemize}
\end{mybox}
\fi

\begin{lemma}{(Completeness)} If the value of $\Upsilon$ is $\ge \alpha$, then the reduction gives an instance of \qpinter{} with optimum value $\ge \alpha$.
\end{lemma}
\ifnum \full=1
\begin{proof}
Consider an assignment to $\Upsilon$ of value $\alpha$ and for each $u$ set $f_u$ to be the corresponding dictator (or $f_u =0$ if $u$ is assigned $\bot$). This gives a ratio at least $\alpha$ (the $L(u)$ terms contribute zero for each $u$).
\end{proof}
\fi
\ifnum\full<1
\vspace{-0.05in}
\fi
\begin{lemma}{(Soundness)}\label{lem:red:soundness}
Suppose the \qpinter{} instance obtained from a reduction (starting with $\Upsilon$) has value $\tau$, then there exists a solution to $\Upsilon$ of value $\ge \tau^2/C$, for an absolute constant $C$.
\end{lemma}
\ifnum\full=1
\begin{proof}
Consider an optimal solution to the instance $\cQ$ of \qpinter{}, and suppose it has a value $\tau >0$.  Since the UG instance is regular, we have
\begin{equation}\label{eq:val-expression}
val(\cQ) = \frac{\sum_u \E_{v \in \Gamma(u)} T_{uv} - \eta \sum_u L(u)}{\sum_u \norm{f_u}_1}.
\end{equation}

First, we move to a solution such that the value is at least $\tau/2$, and for every $u$, $|f_u|_1$ is either zero, or is ``not too small''.  The choice of $\eta$ will then enable us to conclude that each $f_u$ is `almost linear' (there are no higher level Fourier coefficients).
\fi
\ifnum\full<1
\vspace{-0.05in}
See full version for proofs.  We give a sequence of lemmas which imply the Soundness lemma.
\fi
\begin{lemma}\label{lem:qpinter:prelim}
There exists a solution to $\cQ$ of value at least $\tau/2$ with the property that for every $u$, either $f_u =0$ or $\norm{f_u}_1 > \frac{\tau}{n 2^{2R}}$.
\end{lemma}
\ifnum\full=1
\begin{proof}
Let us start with the optimum solution to the instance.  First, note that $\sum_u \norm{f_u}_1 \ge 1/2^R$, because if not, $|f_u(x)| < 1$ for every $u$ and $x \in \{-1,1\}^R$.  Thus if we scale all the $f_u$'s by a factor $z>1$, the numerator increases by a $z^2$ factor, while the denominator only by $z$; this contradicts the optimality of the initial solution.  Since the ratio is at least $\tau$, we have that the numerator of~\eqref{eq:val-expression} (denoted $\num$)  is at least $\tau /2^R$. 

Now since $|\ff_u(S)| \le \norm{f_u}_1$ for any $S$, we have that for all $u,v$, $T_{uv} \le R \cdot \norm{f_u}_1 \norm{f_v}_1$.  Thus $\E_{v \in \Gamma(u)} T_{uv} \le R \cdot \norm{f_u}_1$.  Thus the contribution of $u$ s.t. $\norm{f_u}_1 < \tau/(n 2^{2R})$ to $\num$ is at most $n \times R \cdot \frac{\tau}{n 2^{2R}} < \frac{\tau}{2^{R+1}} < \num/2$.  Now setting all such $f_u =0$ will only decrease the denominator, and thus the ratio remains at least $\tau/2$.  [We have ignored the $L(u)$ term because it is negative and only improves when we set $f_u=0$.]
\end{proof}
\fi
For a boolean function $f$, we define the `linear' and the `non-linear' parts to be
\[ \flin := \sum_i \ff(i) \chi(\{i\}) \quad \text{ and } \fnlin := f - \flin = \sum_{|S| \neq 1} \ff(S) \chi(S).\]
\ifnum\full=1
Our choice of $\eta$ will be such that:
\begin{enumerate}
\item For all $u$ with $f_u \ne 0$, $\norm{\fnlin_u}_2^2 \le \norm{f_u}_1^2/10^6$.  Using Lemma~\ref{lem:qpinter:prelim} (and the na\"ive bound $\tau \ge 1/n$), this will hold if $\eta > 10^6 n^7 2^{4R}$. [A simple fact used here is that $\sum_u \E[T_{uv}] \le nR$.]
\item For each $u$, $\norm{\fnlin_u}_2^2 < \frac{1}{2^{2R}}$.  This will hold if $\eta > n 2^{2R}$ and will allow us to use Lemma~\ref{lem:abhs:cool}.
\end{enumerate}

Also, since by Cauchy-Schwarz inequality, $|f_u|_2^2 \ge \del_u^2$, we can conclude that `most' of the Fourier weight of $f_u$ is on the linear part for {\em every} $u$. We now show that the Cauchy-Schwarz inequality above must be tight up to a constant (again, for every $u$).
\fi
\ifnum\full<1 
The choice of $\eta$ will ensure that for {\em each} $u$, $\norm{\fnlin_u}_2^2$ is {\em tiny}.
\fi
The following is the key lemma in the reduction:  it says that if a boolean function $f$ is `nearly linear', then it must also be {\em spread out} [which is formalized by saying $\norm{f}_2 \approx \norm{f}_1$].  This helps us deal with the main issue in a reduction with a ratio objective -- showing we cannot have a large numerator along with a very small value of $\norm{f}_1$ (the denominator).  Morally, this is similar to a statement that a boolean function with a {\em small support} cannot have all its Fourier mass on the {\em linear} Fourier coefficients.

\begin{lemma} \label{lem:smallball}
Let $f : \{-1,1\}^R \mapsto [-1,1]$ satisfy $\norm{f}_1 = \delta$.  Let $\flin$ and $\fnlin$ be defined as above. Then if $\norm{f}_2^2 > (10^4+1) \delta^2$, we have $\norm{\fnlin}_2^2 \ge \delta^2$.
\end{lemma}
\ifnum \full=1
\begin{proof}
Suppose that $\norm{f}_2^2 > (10^4+1) \delta^2$, and for the sake of contradiction, that $\norm{\fnlin}_2^2 < \delta^2$. Thus since $\norm{f}_2^2 = \norm{\flin}_2^2 + \norm{\fnlin}_2^2$, we have $\norm{\flin}^2 > (100 \delta)^2$.

If we write $\alpha_i = \widehat{f}(\{i\})$, then $\flin(x) = \sum_i \alpha_i x_i$, for every $x \in \{-1,1\}^R$.  From the above, we have $\sum_i \alpha_i^2 > (100\delta)^2$.  Now if $|\alpha_i| > 4\delta$ for some $i$, we have $\norm{\flin}_1 > (1/2) \cdot 4\delta $, because the value of $\flin$ at one of $x, x \oplus e_i$ is at least $4\delta$ for every $x$.  Thus in this case we have $\norm{\flin}_1 > 2\del$.

Now suppose $|\alpha_i| < 4 \delta$ for all $i$, and so we can use Lemma~\ref{lem:berry-esseen} to conclude that $  \Pr_x (\flin(x) > 100\delta/10) \ge 1/4$, which in turn implies that $|\flin|_1 > (100\delta / 10) \cdot \Pr_x (\flin(x) > 100\delta/10) > 2\delta$.

Thus in either case we have $\norm{\flin}_1>2\del$.  This gives  $\norm{f - \flin}_1 > \norm{\flin}_1 - \norm{f}_1 > \delta$, and hence $\norm{f- \flin}_2^2 > \delta^2$ (Cauchy-Schwarz), which implies $\norm{\fnlin}_2^2 > \delta^2$, which is what we wanted.
\end{proof}

Now, let us denote $\del_u = |f_u|_1$. Since $\Upsilon$ is a unique game, we have for every edge $(u,v)$ (by Cauchy-Schwarz),
\begin{equation}
T_{u,v} =  \sum_i \ff_u(\{i\})\ff_v(\{ \pi_{uv}(i)\}) \leq \sqrt{\sum_i \ff_u(\{i\})^2}\sqrt{\sum_j \ff_u(\{j\})^2} \leq |f_u|_2 |f_v|_2  \label{eq:tuv:ub}
\end{equation}
Now we can use Lemma~\ref{lem:smallball} to conclude that in fact, $T_{u,v} \le 10^4 \del_u \del_v$.  Now consider the following process: while there exists a $u$ such that $\del_u >0$ and $\E_{v \in \Gamma(u)} \del_v < \frac{\tau}{4\cdot 10^4}$, set $f_u=0$.  We claim that this process only increases the objective value.  Suppose $u$ is such a vertex.  From the bound on $T_{uv}$ above and the assumption on $u$, we have $\E_{v \in \Gamma(u)} T_{uv} < \del_u \cdot \tau/4$.  If we set $f_u=0$, we remove at most twice this quantity from the numerator, because the UG instance is regular [again, the $L(u)$ term only acts in our favor].  Since the denominator reduces by $\del_u$, the ratio only improves (it is $\ge \tau/2$ to start with).

Thus the process above should terminate, and we must have a non-empty graph at the end.  Let $S$ be the set of vertices remaining.  Now since the UG instance is regular, we have that $\sum_u \del_u = \sum_u \E_{v \in \Gamma(u)} \del_v$.  The latter sum, by the above is at least $|S| \cdot \tau / (4 \cdot 10^4)$.  Thus since the ratio is at least $\tau/2$, the numerator $\num \ge |S| \cdot \frac{\tau^2}{8 \cdot 10^4}$.

Now let us consider the following natural randomized rounding: for vertex $u \in S$, assign label $i$ with probability $|\ff_u(\{i\})| / (\sum_i |\ff_u(\{i\})|)$.  Now observing that $\sum_i |\ff_u(\{i\})| <2$ for all $u$ (Lemma~\ref{lem:abhs:cool}), we can obtain a solution to ratio-UG of value at least $\num / |S|$, which by the above is at least $\tau^2/ C$ for a constant $C$.

This completes the proof of Lemma~\ref{lem:red:soundness}.
\end{proof}
\fi
\ifnum \full<1
This now can be used to prove the soundness:  the numerator can be bounded in terms of the $\norm{f_u}_2$, which can now be bounded in terms of $\norm{f_u}_1$.  This allows us to restrict to $u$ with $\norm{f_u}_1 \ge \tau/C$, and then a randomized rounding gives a solution of value $\tau^2/C$.  [In the process, we use crucially the fact that we started off with a {\em unique game} instance, and that it is regular. See the full version for details.]
\fi
\bibliographystyle{alpha}
\bibliography{qp-ratio}
\ifnum \full<1
We append the full version of the paper from the following page.
\fi
\ifnum\full=1
\appendix

\section{Hardness of QP-Ratio}
\subsection{Boolean analysis}\label{app:hardness:bool}
\begin{lemma}\cite{qp-hard} \label{lem:abhs:cool}
Let $f_u : \{-1,1\}^R \rightarrow [-1,1]$ be a solution to $\cQ$ of value $\tau >0$. Then 
\[\forall u \in V \qquad \sum_{i=1}^{R} |\ff_u(\{i\})| \leq 2. \]
\end{lemma}
\begin{proof}
Assume for sake of contradiction that $\sum_i \ff_u(\{i\}) > 2$. \\
Since $\flin_u$ is a linear function with co-efficients $\{\ff_u(\{i\})\}$, there exists some $y \in \{-1,1\}^R$ such that 
$\flin_u(y)=\sum_i |\ff_i(\{i\})| >2$. For this $y$, we have $\fnlin(y)=f(y)-\flin(y) < -1$.

Hence $|\fnlin|_2^2 > 2^{-R}$, which gives a negative value for the objective, for our choice of $\eta$.
\end{proof}

The following is the well-known Berry-Ess\'een theorem (which gives a
quantitative version of central limit theorem).  The version below is from~\cite{ryan}.
\begin{lemma}\label{lem:berry-esseen}
Let $\alpha_1, \dots, \alpha_R$ be real numbers satisfying $\sum_i \alpha_i^2 =1$, and $\alpha_i^2 \le \tau$ for all $i \in [R]$. Let $X_i$ be i.i.d. Bernoulli ($\pm 1$) random variables. Then for all $\theta >0$, we have
\[ \big| \Pr[\sum_i \alpha_i X_i >\theta] - N(\theta) \big| \le \tau, \]
where $N(\theta)$ denotes the probability that $g > \theta$, for $g$ drawn from the univariate Gaussian $\mathcal{N}(0,1)$.
\end{lemma}

\subsection{Reducing \qpinter{} to \qpr{}}\label{app:qpinter-qpr}
In this section we will prove Lemma~\ref{lem:qpinter-qpr}.  Let us start with a simple observation
\begin{lemma} \label{lem:app:xmin}
Let $A$ be an $n \times n$ matrix (it could have arbitrary diagonal entries). Suppose $\{x_i\}$, $1 \le i \le n$ is the optimum solution to
\[ \max_{x_i \in [-1,1]} \frac{x^T Ax}{\sum_i |x_i|}. \]
Now let $\del < \min\{\frac{\eps}{2\norm{A}_1}, \frac{1}{2n} \}$ (where $\norm{A}_1 = \sum_{i,j} |a_{ij}|$).  Then perturbing each $x_i$ additively by $\del$ (arbitrarily) changes the value of the ratio by an additive factor of at most $\eps$.
\end{lemma}
\begin{proof}
First note that $\sum_i |x_i| \ge 1$, because otherwise we can scale all the $x_i$ by a factor $z>1$ and obtain a feasible solution with a strictly better value [because the numerator scales by $z^2$ and the denominator only by $z$].  Thus changing each $x_i$ by $\del < 1/2n$ will keep the denominator between $1/2$ and $3/2$.  Now consider the numerator: it is easy to see that a term $a_{ij} x_i x_j$ changes by at most $\del |a_{ij}|$, and thus the numerator changes by at most $\del \norm{A}_1$.  Thus the ratio changes by an additive factor at most $2\del \norm{A}_1 < \eps$, by the choice of $\del$.
\end{proof}

\begin{proof}[Proof of Lemma~\ref{lem:qpinter-qpr}]
Start with an instance of \qpinter{} given by $A_{(n \times n)}$, and suppose the optimum value is $\opt_1$.  Let $m$ be an integer which will be chosen later [think of it as sufficiently large].  Consider the quadratic form $B$ on $n \cdot m$ variables, defined by writing $x_i = \frac{1}{m} \cdot (y_i^{(1)} + y_i^{(2)} + \dots + y_i^{(m)})$, and expanding out $x^T Ax$.  Let $C$ be a form on (the same) $n\cdot m$ variables obtained from $B$ by omitting the square (diagonal) terms.
Now consider the \qpr{} instance given by $C$. That is,
\begin{equation}\label{eq:diag:optproblem} \text{maximize } \frac{y^T Cy}{\frac{1}{m}\sum_{i,j} |y_i^{(j)}|} \text{, subject to } y_i^{(j)} \in \{-1, 0, 1\}.
\end{equation}
Let us write $a_{ii} = -\alpha_i$ (by assumption $\alpha_i \ge 0$).  The we observe that
\begin{equation}
\frac{y^T Cy}{\frac{1}{m} \sum_{i,j} |y_i^{(j)}|} = \frac{y^T By}{\frac{1}{m} \sum_{i,j} |y_i^{(j)}|} + \frac{\sum_i \frac{\alpha_i}{m^2} \cdot \big(\sum_j |y_i^{(j)}| \big) }{\frac{1}{m}\sum_{i,j} |y_i^{(j)}|}
\end{equation}
By the assumption on $\alpha_i$, we have
\begin{equation}\label{eq:opt-sandwich}
\frac{y^T By}{\frac{1}{m} \sum_{i,j} |y_i^{(j)}|} ~\le ~\frac{y^T Cy}{\frac{1}{m} \sum_{i,j} |y_i^{(j)}|} ~\le~ \frac{y^T By}{\frac{1}{m} \sum_{i,j} |y_i^{(j)}|} + \sum_i \frac{\alpha_i}{m}
\end{equation}
We prove the two inequalities separately.  First let us start with an optimum solution $\{x_i\}$ to \qpinter{} (from $A$) with value $\opt_1$.  As above, define $\del = \min\{\frac{\eps}{2\norm{A}_1}, \frac{1}{2n} \}$.  Let us round the values $x_i$ to the nearest integer multiple of $\del$ [for simplicity we will assume also that $1/\del$ is an integer].  By Lemma~\ref{lem:app:xmin}, this will change the objective value by at most $\eps$.  We will choose $m$ to be an integer multiple of $1/\del$, thus if we set $y_i^{(j)} = \emph{sign}(x_i)$ for $j=1,2,\dots,m x_i$ and $0$ for the rest, we obtain a value at least $\opt_1 - \eps$ for the $\qpr{}$ problem defined by $C$ [using the first half of \eqref{eq:opt-sandwich}].  Thus $\opt_2 \ge \opt_1-\eps$.

Now consider a solution to the $\qpr{}$ problem defined by $C$, and set $x_i = \frac{1}{m} \cdot (y_i^{(1)} + \dots y_i^{(m)})$.  For this assignment, we have
\[ \frac{x^T A x}{\sum_i |x_i|} = \frac{y^T By}{\sum_i \frac{1}{m} (\sum_{j} |y_i^{(j)}|)} \ge \frac{y^T By}{\frac{1}{m} \sum_{i,j} |y_i^{(j)}|} \ge \frac{y^T Cy}{\frac{1}{m} \sum_{i,j} |y_i^{(j)}|} - \eps, \]
because we will choose $m \ge \frac{\norm{A}_1}{\eps}$.  This implies that $\opt_1 \ge \opt_2 - \eps$.

Thus we need to choose $m$ to be the smallest integer larger than $\max\{\frac{2\norm{A}_1}{\eps}, 2n\}$ for all the bounds to hold.  This gives the claimed bound on the size of the instance.
\end{proof}

\subsection{Towards NP-hardness -- LabelCover with SSE} \label{app:towardsnp}
\newcommand{\lc}{\textsf{Label Cover}~}
The PCP theorem, \cite{AroraS,AroraStar} combined with the parallel
repetition theorem~\cite{Raz} yields the following theorem.

\begin{theorem}[\lc hardness] There exists a constant $\gamma > 0$ so
  that any $3$-SAT instance $w$  and any $R > 0$, one can
  construct a \lc instance $\cL$, with $|w|^{O(\log R)}$ vertices,
  and label set of size $R$, so that: if $w$ is satisfiable,
  $\val(\cL) = 1$ and otherwise $\val(\cL)=\tau < R^{-\gamma}$.  Further,
  $\cL$ can be constructed in time polynomial in its size.
\end{theorem}

\begin{definition}[$(\mu, \eta)$-Expanding \lc] An instance of the \lc problem
  is said to be $(\mu, \eta)$-expanding if for every $\mu' < \mu$,
  and functions $f: \cA \to [0, 1]$ and $g: \cB \to [0, 1]$ such that
  $E_{a \in \cA} [f(a)] = E_{b \in \cB} [g(b)] = \mu'$,
  \[ E_{(a, b) \in E} [f(a)\ g(b)] \le \mu' \eta  .\]
\end{definition}

\begin{theorem} \label{thm:ratiolc}
For every $\delta > 0$ and $\eta>\delta^{1/3}$, one can convert a \lc instance $\cL$
  into an instance $(\delta,\eta)$ Expanding \lc $\cL'$ (with the same completeness vs soundness). Further, the size of $\cL'$ is at most $|\cL|^{1/\delta}$ and has
  a label set of size at most $R/\delta$.
\end{theorem} 

\begin{proof}
We use one of product instances used in \cite{kh} (which gives Mixing but not Smoothness - Appendix A.2). 
They can argue about the expansion of only sets that are sufficiently large (constant fraction), while we need to work with 
all set sizes $< \delta n$.

 Given an instance of \lc represented as $\Upsilon =
  (\cA\cup \cB, E, \Pi, [R])$, let $\Upsilon^k = (\cA^k \cup \cB, E^k,
  \Pi, [R])$, where $(1/\eta)^3 < k < 1/\delta$. Let $g: \cB \to [0, 1]$ be any function defined on the
  right hand side.  For each $a \in \cA$, let $X_a$ denote the average
  value of $g$ over the neighborhood of $a$.  Since $\Upsilon$ is
  right-regular, we have $\E_{a \in \cA} [ X_a ] = \E_{b \in \cB}
  g(b)$.

  For a vertex $a = (a_1, a_2, \ldots a_k) \in \cA^k$, define $Y_a$ to
  be the average value of $g$ over the neighborhood of $a$ (counting
  multiple edges multiple times).  

  \begin{align*}
    \E_{a \in \cA^k} [ Y_a^2 ] &= \E_{a_1, a_2, \ldots, a_k \in \cA} \left[
      \left( \frac{\sum_i X_{a_i}}{k} \right)^2 \right]\\
    &= \frac{1}{k} \E_{a \in \cA} [ X^2_{a_i} ] + \frac{k^2 - k}{k^2}
    \E_{a, a' \in \cA} [ X_a X_{a'} ]\\
    &\le \frac{1}{k} \delta + \delta^2
    \le \frac{2 \delta}{k} && \text{ (for $\delta < \nfrac{1}{k}$) }
  \end{align*}

  Thus, by a second moment bound, the fraction of vertices in $\cA^k$
  with $Y_a$ greater than $\eta$ is at most $\nfrac{2
    \delta}{k\eta^2}$.  Thus, for any $f: \cA^k \to [0, 1]$, such that
  $\E[f] = \delta$,
\begin{align*}
  \E_{(a,b)\in E} [ f(a) g(b) ] &\le \frac{2 \delta}{k\eta^2} + \eta \left(\delta -\frac{2\delta}{k\eta^2}\right) \\
  &\le 3 \delta \eta \qquad \text{since }(1/\eta)^3 < k 
\end{align*}
\end{proof}

\paragraph{Ratio Label Cover.}
Consider the ratio version of Label Cover -- the goal is to find a partial assignment to a Label Cover instance, which maximizes the fraction of edges satisfied, divided by the fraction of vertices which have been assigned labels (an edge $(u,v)$ is satisfied iff both the end points are assigned labels which satisfy the constaint $\pi_{u,v}$).

It follows by a fairly simple argument that Theorem~\ref{thm:ratiolc} shows that Ratio Label Cover is NP-hard to approximate within any constant factor. We sketch an argument that shows NP-hardness of $1$ vs $\gamma$ for any constant $\gamma>0$. We start with Label Cover instance $\Upsilon$ with completeness $1$ and soundess $\tau<\gamma^{4/3}$ (and appropriate label size). Our Ratio Label Cover instance is essentially obtained by applying Theorem~\ref{thm:ratiolc} : the instance has soundness $\tau$ and is $(\gamma^{1/3},\gamma)$ expanding. The completeness of this instance is easily seen to be $1$.

To argue the soundess, it first suffices to only consider solutions $f,g$ which have equal measure (similar to the argument in section~\ref{sec:rand-k-and} -- the instance first has its right side duplicated so that right and left sizes are equal, and then upto a factor $2$ loss, we only pick equal number of vertices on both sides). The expansion of sets of measure $<\delta$ have value at most $\gamma$ due to the expansion of the instance. Suppose there is a solution of measure $\ge \delta$ which has ratio label cover value $\gamma$, then we obtain a solution to the Label Cover instance $\Upsilon$ of value at least $\gamma \cdot \delta =\gamma^{4/3}>\tau$, which is a contradiction.

\section{APX-hardness of \qpr{}} \label{app:np-hardness}
We proved that \qpr{} is hard to approximate to an $O(1)$ factor assuming the small-set expansion conjecture. Here we prove a weaker hardness result -- that there is no PTAS -- assuming just $\mathbf{P} \neq \mathbf{NP}$. We do not go into the full details, but the idea is the following.

We reduce Max-Cut to an instance of \qpr{}. The following is well-known (we can also start with other QP problems instead of Max-Cut)
\begin{quote}
There exist constants $\nfrac{1}{2} < \rho' < \rho$ such that: given a graph $G = (V,E)$ which is regular with degree $d$, it is NP-hard to distinguish between\\
{\sc Yes.} MaxCut$(G) \ge \rho \cdot \nfrac{nd}{2}$, and\\
{\sc No.} MaxCut$(G) \le \rho' \cdot \nfrac{nd}{2}$.
\end{quote}

Given an instance $G=(V,E)$ of Max-Cut, we construct an instance of \qpr{} which has $V$ along with some other vertices, and such that in an OPT solution to this \qpr{} instance, {\em all} vertices of $V$ would be picked (and thus we can argue about how the best solution looks).

First, let us consider a simple instance: let $abcde$ be a 5-cycle, with a cost of $+1$ for edges $ab, bc, cd, de$ and $-1$ for the edge $ae$. Now consider a \qpr{} instance defined on this graph (with $\pm 1$ weights). It is easy to check that the best ratio is obtained when precisely four of the vertices are given non-zero values, and then we can get a numerator cost of $3$, thus the optimal ratio is $3/4$.

Now consider $n$ cycles, $a_i b_i c_i d_i e_i$, with weights as before, but scaled up by $d$. Let $A$ denote the vertex set $\{a_i\}$ (similarly $B, C, ..$). Place a clique on the set of vertices $A$, with each edge having a cost $10 d/n$. Similarly, place a clique of the same weight on $E$. Now let us place a copy of the graph $G$ on the set of vertices $C$.

It turns out (it is actually easy to work out) that there is an optimal solution with the following structure: (a) all $a_i$ are set to $1$, (b) all $e_i$ are set to $-1$ (this gives good values for the cliques, and good value for the $a_i b_i$ edge), (c) $c_i$ are set to $\pm 1$ depending on the structure of $G$, (d) If $c_i$ were set to $+1$, $b_i = +1$, and $d_i =0$; else $b_i=0$ and $d_i = -1$ (Note that this is precisely where the $5$-cycle with one negative sign helps!)

Let $x_1, ..., x_n \in \{-1,1\}$ be the optimal assignment to the Max-Cut problem. Then as above, we would set $c_i = x_i$. Let the cost of the MaxCut solution be $\theta \cdot \nfrac{nd}{2}$. Then we set $4n$ of the $5n$ variables to $\pm 1$, and the numerator is (up to lower order terms):
\[ 2 \cdot (10d/n) \nfrac{n^2}{2} + \theta \cdot \nfrac{nd}{2} + 3nd = (\Delta + \theta)nd, \]
where $\Delta$ is an absolute constant.

We skip the proof that there is an optimal solution with the above structure. Thus we have that it is hard to distinguish between a case with ratio $(\Delta + \rho')d/4$, and $(\Delta + \rho)d/4$, which gives a small constant factor hardness.

\section{Reduction from a weighted to an Unweighted version}\label{app:kand:instance}
\newcommand{\Dom}{\{-1,0,1\}}
\begin{lemma}
Let $A$ be an $n \times m$ matrix, and let $w \ge 1$ be an integer.  Let $opt_1$ denote the optimum value of the problem
\[ \max_{x \in \Dom^n, y \in \Dom^m} \frac{x^T A y}{w \norm{x}_1 + \norm{y}_1}. \]
Let $B$ be a $wn \times m$ matrix formed by placing $w$ copies of $A$ one below the other.  [In terms of bipartite graphs, this just amounts to making $w$ copies of the left set of vertices.]  Let $opt_2$ denote the optimum value of the problem
\[ \max_{z \in \Dom^{wn}, y \in \Dom^m} \frac{z^T B y}{\norm{z}_1 + \norm{y}_1}.\]
Then $opt_1 = opt_2$.
\end{lemma}
\begin{proof}
It is clear that $opt_2 \ge opt_1$: simply take a solution of value $opt_1$ for the first problem and form $z$ by taking $w$ copies of $x$.  To see the other direction, let us view $z$ as being formed of `chunks' $z_1, \dots, z_w$ of size $n$ each.  Consider a solution to the second problem of value $opt_2$.  Then
\[ opt_2 = \frac{ z_1^T A y + z_2^T Ay + \dots z_w^T  A y}{\norm{z_1}_1 + \dots \norm{z_w}_1 + |y|}, \]
which implies that if we set $x = z_i$ which gives the largest value of $z_i^TA y / \norm{z_i}_1$, we obtain a value at least $opt_2$ to the first problem.

This completes the proof.
\end{proof}

\fi
\end{document}